\documentclass[journal,twoside,web]{ieeecolor}
\usepackage{setspace}

\usepackage{generic}
\usepackage{algorithmic}
\usepackage[noadjust]{cite}
\usepackage{amsmath,amssymb,amsfonts}
\usepackage{algorithmic}
\usepackage{graphicx}
\usepackage{textcomp}
\usepackage{dsfont}
\usepackage{subcaption}
\usepackage[font=small]{caption}
\usepackage{mathbbol} 
\let\labelindent\relax
\usepackage{enumitem}
\newtheorem{theorem}{Theorem}[section]
\newtheorem{corollary}{Corollary}[section]
\newtheorem{lemma}[theorem]{Lemma}

\newtheorem{proposition}[theorem]{Proposition}
\newtheorem{assumption}{Assumption}
\newtheorem{definition}{Definition}
\newtheorem{example}{Example}[section]
\newcommand{\bs}[1]{\boldsymbol{#1}}

\newcommand{\Na}{{N_{\edits{\rm a}}}}

\newcommand{\No}{{N_{\edits{\rm o}}}}

\newcommand{\Zz}{\bs{Z}}

\newcommand{\R}[1]{\mathbb{R}^{#1}}

\newcommand{\edits}[1]{{\color{black} #1}}

\usepackage{dsfont}

\def\BibTeX{{\rm B\kern-.05em{\sc i\kern-.025em b}\kern-.08em
    T\kern-.1667em\lower.7ex\hbox{E}\kern-.125emX}}
\begin{document}

\title{Multi-topic belief formation through bifurcations over signed social networks}
\author{Anastasia Bizyaeva, \IEEEmembership{Member, IEEE}, Alessio Franci, and Naomi E. Leonard, \IEEEmembership{Fellow, IEEE}
\thanks{This work was supported in part by ONR grant N00014-19-1-2556, ARO grant W911NF-18-1-0325, and by  
NSF Graduate Research Fellowship DGE-2039656.}
\thanks{A.B.  with  NSF AI Institute in Dynamic Systems and  Dept. of Mechanical Engineering, University of Washington, Seattle, WA, 98195 USA;  anabiz@uw.edu. 
A.F. with  Dept. of Electrical Engineering and Computer Science at the University of Liège, Belgium, and with the WEL Research Institute, Wavre, Belgium;  afranci@uliege.be.
N.E.L. with  Dept.
of Mechanical and Aerospace Engineering, Princeton University, Princeton,
NJ, 08544 USA;  naomi@princeton.edu.}
}

\maketitle

\begin{abstract}
We propose and analyze a nonlinear dynamic model of continuous-time multi-dimensional belief formation over signed social networks. Our model accounts for the effects of a structured belief system, self-appraisal, internal biases, and various sources of cognitive dissonance posited by recent theories in social psychology. We prove that \edits{agents become opinionated} as a consequence of a bifurcation. We analyze how the balance of social network effects in the model controls the nature of the bifurcation and, therefore, the belief-forming limit-set solutions. Our analysis provides constructive conditions on how multi-stable network belief equilibria and belief oscillations emerging at a belief-forming bifurcation depend on the communication network graph and belief system network graph. Our model and analysis provide new theoretical insights on the dynamics of social systems and a new principled framework for designing decentralized decision-making on engineered networks in the presence of structured relationships among alternatives.
\end{abstract}


\section{Introduction}

The study of belief dynamics has a long history of inquiry across disciplines.
In social science, belief formation models are built to explore belief change in individuals and collective phenomena including the spread of social norms,  outcomes of elections, and  emergence of consensus or polarization in societies. In the study of collective animal behavior, similar models are deployed to understand the outcomes of social decisions in animal groups.
Belief formation models are also prominent in control systems and robotics, where they inspire the design of new protocols for networked consensus formation, multi-agent decision-making, navigation, task scheduling, control of networks, and human-robot interaction \cite{olfati2004consensus,montes2010opinion,montes2011majority,altafini2012consensus,HaiminCDC2023,cathcart2022opinion,leung2023leveraging,paine2023model}. 

Recent work in social psychology points out an important limitation of belief formation models \cite{galesic2021human,dalege2023networks}. Existing models either focus on the internal dynamics of beliefs of individuals without accounting for social network effects or they focus on the dynamics of social networks without incorporating meaningful internal factors, such as the nature of interpersonal relationships and the effects of cognitive dissonance. Emerging evidence suggests instead that belief formation should be viewed through a framework that synthesizes internal and external effects, with individual beliefs that are ``embedded in a multidimensional, self-sustaining system of mental representations and shaped and reinforced continuously in the social interactions people have in their communities'' \cite{vlasceanu2022network}. Brain imaging studies also provide evidence that social decisions about alternatives are made ``through integrating multiple types of inference about oneself, others, and environments, processed in distinct brain modules'' \cite{suzuki2015neural}.

A similar shift in perspective is happening in the collective animal behavior literature. Traditionally, studies have focused on the role of external interaction networks in shaping the behavior of the group \cite{strandburg2013visual,rosenthal2015revealing}. In contrast, recent studies present evidence that internal cognitive processes play a critical role in shaping the outcomes of animal collective decisions, e.g., during spatial movement \cite{sridhar2021geometry,oscar2023simple}. These studies suggest that realistic models of collective animal behavior should account for the internal dynamics of cognitive representations in addition to the external social and environmental factors.

In this paper we take a step forward towards a mathematical understanding of the interplay between internal and external factors in shaping collective belief formation. To do this \edits{we present and analyze a nonlinear dynamic model of collective belief formation on multiple alternatives that generalizes the opinion dynamics model of \cite{bizyaeva2022},\cite{franci2022breaking}}. To account for external interactions our model includes a social network on which agents can have cooperative and antagonistic relationships. To account for internal factors in individual cognition our model incorporates internal biases and networked relationships between internal belief representations. The coupling between social relationships and a structured belief system is the source of various forms of cognitive dissonance in our model.

Formal models of belief formation outside of the statistical physics literature are predominantly linear. Classically, belief formation is modeled as a linear discrete-time DeGroot consensus process in which cooperative agents update scalar beliefs via local averaging \cite{degroot1974reaching}. Common variants of consensus dynamics include its continuous-time extension and the inclusion of time delays, time-varying communication networks, and antagonistic relationships between agents \cite{olfati2004consensus,altafini2012consensus}. To account for multiple simultaneously evolving beliefs, a number of multidimensional consensus models have recently been explored in discrete-time \cite{friedkin2016network,parsegov2016novel,ye2019consensus,he2022opinion} and continuous-time \cite{pan2018bipartite,nedic2019graph,ye2020continuous,ahn2020opinion,Liu2021interplay,wang2022characterizing,he2023opinion}. Unlike their scalar counterparts, these models explicitly consider the simultaneous evolution of multiple beliefs within the agents, and thereby account for some internal interdependence of belief representations. However, most existing  multi-belief models are still linear and grounded in the local averaging assumption of DeGroot. A number of nonlinear models of belief or opinion dynamics have also been considered in the literature, e.g. \cite{lorenz2007continuous,dandekar2013biased,Gray2018,xia2020analysis,Liu2021interplay,Wang2022,mei2022micro,shrinate2022nonlinear}. However, these models, excepting \cite{Liu2021interplay}, only account for scalar beliefs and do not capture rich internal systems of belief representations.

There is compelling evidence that realistic belief formation dynamics are nonlinear, which motivates our use of a multidimensional nonlinear model. At the individual level, nonlinear processing of perceptual evidence in the brain 
is important to the formation of beliefs and decision-making  
\cite{UM2001,bogacz2007extending}. At the group level, a recent analysis \edits{of} online social network \edits{data} suggests that beliefs cluster around several isolated, simultaneously stable equilibria, a dynamical behavior that cannot be reproduced by linear models \cite{introne2023measuring}. The formation of beliefs and decisions in isolated individuals and in social networks is  often linked to phase transitions in computational statistical physics (i.e. stochastic) models \cite{rodriguez2016collective,baumann2020modeling,siegenfeld2020negative,baumann2021emergence,sridhar2021geometry,brandt2021evaluating,gajewski2022transitions,li2022modeling,oscar2023simple,ojer2023modeling}. A deterministic analogue of a phase transition is a \textit{bifurcation}, i.e. a parameterized change in the qualitative behavior of a dynamical system. Bifurcations are nonlinear phenomena \edits{and are 
central to} our analysis of belief formation dynamics.

The following are the primary contributions of this work. 1) We \edits{adapt} a nonlinear multi-dimensional model of belief dynamics \edits{recently introduced in \cite{bizyaeva2022} to} account for structured relationships between options \edits{and heteoreneous signed social relationships}. We connect this model and its analysis to \edits{Networks of Beliefs Theory} recently proposed in social psychology \cite{dalege2023networks}. 2) We prove that agents \edits{become opinionated} when the model exhibits a bifurcation. The emergent beliefs on the network post-bifurcation are shaped by properties of the communication graph between agents, their shared belief system, and the balance of various social effects in the model. 3)  We prove sufficient conditions for the emergence of multi-stable belief equilibria and of belief oscillations, and describe the agents' beliefs in terms of spectral properties of communication and belief system graphs. 

\edits{Our contributions are complementary to 
\cite{bizyaeva2022},\cite{franci2022breaking}, 
which considered all-to-all coupling between belief dimensions and a homogeneously-signed communication graph. We introduce a graph structure between belief dimensions and allow heterogeneously-signed communication and belief graphs.  
}

In Section \ref{sec:background} we present definitions and  background on signed graphs. In Section \ref{sec:model} we define the belief formation model and connect the model to recent work in social psychology. In Section \ref{sec:lin} we establish conditions for the onset of an indecision-breaking bifurcation in the model. In Section \ref{sec:params} we analyze the effect of various model parameters on shaping and controlling solutions at the bifurcation. In Section \ref{sec:pitchfork} we derive necessary conditions for the onset of multiple equilibria in a pitchfork bifurcation, characterize the resulting steady-state solutions, and provide sufficient conditions for graph structures that support a pitchfork bifurcation. In Section \ref{sec:Hopf} we derive necessary conditions for the onset of belief oscillations in a Hopf bifurcation. We conclude in Section \ref{sec:final}. \edits{Unless stated otherwise, proofs are  in the Appendix.}

\section{Background \label{sec:background}}

\subsection{Notation and mathematical preliminaries}

For  $x = a + i b = r e^{i \phi} \in \mathbb{C}$, $\overline{x} = a - i b = r e^{-i \phi}$ is its complex conjugate, $|x| = \sqrt{x \overline{x}} = r$  its modulus, and $\operatorname{arg}(x)$  its argument $\phi$. For vectors $\mathbf{x},\mathbf{y} \in \mathbb{C}^n$, the inner product is $\langle \mathbf{x}, \mathbf{y} \rangle = \mathbf{x} \overline{\mathbf{y}}^T$.  The norm of $\mathbf{x} \in \mathbb{C}^{n}$ is $\| \mathbf{x} \| = \sqrt{\langle \mathbf{x}, \mathbf{x} \rangle}$. For $\mathbf{x} = (x_1, \dots, x_n)\edits{^T} \in \mathbb{R}^n$, the matrix $\operatorname{diag}\{\mathbf{x} \} \in \R{n \times n}$ is a diagonal matrix with $x_i$ in row $i$, column $i$. Define $\mathcal{I}_n \in \R{n \times n}$ as the identity matrix, $\mathbf{0}_n\in \R{n}$ as the zero vector,  and $\mathbf{1}_n \in \R{n}$ as the vector of ones.  We utilize two common matrix products: the Kronecker product, denoted by the symbol $\otimes$, and the element-wise Hadamard product, denoted by $\odot$. 
Given  vectors $\mathbf{v},\mathbf{w}$ or  matrices $M,N$, we say $\mathbf{v} \succ \mathbf{w}$ if $v_i > w_i$ for all $i$ and $M \succ N$ if $M_{ij} > N_{ij}$ for all $i,j$.

For a square matrix $A \in \mathds{R}^{n\times n}$, its spectrum is the collection of eigenvalues $\sigma(A) = \{\lambda_1, \dots, \lambda_{n} \}$ where $\lambda_i A = \lambda_i \mathbf{v}_i$ for some eigenvector $\mathbf{v}_i \in \R{n}$. The spectral radius of $A$ is $\rho(A)= \operatorname{max}\{|\lambda_i|, \ \lambda_i \in \sigma(A)\}$. The kernel of $A$ is $\mathcal{N}(A) = \{\mathbf{v} \in \mathds{R}^n \ \rm{s.t.} \ A \mathbf{v} = \mathbf{0} \}$. Two square matrices $A,B \in \R{n}$ are \textit{co-spectral} if $\sigma(A) = \sigma(B)$. An eigenvalue $\lambda \in \sigma(A)$ is a \textit{leading eigenvalue} of $A$ if $\operatorname{Re}(\lambda) \geq \operatorname{Re}(\mu)$ for all $\mu \in \sigma(A)$. A leading eigenvalue $\lambda$ is a \textit{dominant eigenvalue} if $\lambda = \rho(A)$.  $A$ has the \textit{strong Perron-Frobenius property} if its dominant eigenvalue $\lambda$ is unique, positive, $|\lambda_i| < \lambda$ for all $\lambda_i \neq \lambda$ in $\sigma(A)$, and it has a corresponding eigenvector $\mathbf{v} \succ \mathbf{0}$. $A$  is \textit{irreducible} if it cannot be transformed into an upper triangular matrix through similarity transformations. $A$ is \textit{eventually positive (eventually nonnegative)} if there exists a positive integer $k_0$ such that $A^k \succ 0_{N \times N}$ ($A^k \succeq 0_{N \times N}$) for all integers $k > k_0$. A \textit{permutation matrix} is a square matrix $P_n \in \{0, 1\}^{n \times n}$ formed by permuting the rows of the identity matrix $\mathcal{I}_n$. It has exactly one entry of 1 in each row and each column, with all other entries being zero, and it is an orthonormal matrix that satisfies $P_n P^T_n = P^T_n P_n = \mathcal{I}_n$.

\begin{proposition}[\textit{matrices with Perron-Frobenius property}]
 i) \cite[Theorem 2.2]{noutsos2006perron} The following statements are equivalent for a real square $n \times n$ matrix $A$: (1) $A$ and $A^T$ have the strong Perron-Frobenius property; (2) $A$ is eventually positive; (3) $A^T$ is eventually positive.
 ii) \cite[Theorem 8.4.4]{horn2012matrix} Suppose $A$ is a real, square, irreducible matrix with nonnegative entries, $A \succeq 0_{n \times n}$. Then $A$ has a unique dominant eigenvalue $\lambda = \rho(A)$ with a corresponding eigenvector $\mathbf{v} \succ \mathbf{0}$.
    \label{prop:PerFr} 
\end{proposition}

\subsection{Signed graphs} The discussion of signed graphs presented here adapts the conventions from \cite{zaslavsky2013matrices,altafini2012consensus,belardo2019open}. A \textit{signed graph} $\mathcal{G} = (\mathcal{V}, \mathcal{E}, s)$ consists of a set of \textit{vertices} $\mathcal{V} = (1, \dots, n)$, a set of \textit{edges} $\mathcal{E} \subseteq \mathcal{V}\times \mathcal{V}$, and a \textit{signature function} $s: \mathcal{E} \to \{-1,1\}$. An edge $(i,j) \in \mathcal{E}$ begins at vertex $i$ and ends at vertex $j$. The entries of the \textit{signed adjacency matrix} $A \in \R{n\times n}$ corresponding to $\mathcal{G}$ are $a_{ij} = s(i,j)$ if $(i,j) \in \mathcal{E}$ and $a_{ij} = 0$ otherwise. All graphs in this paper are \textit{simple}, i.e. contain no self-loops $(i,i) \not\in \mathcal{E}$ for all $i \in \mathcal{V}$ and contain at most one edge $(i,j)$ that begins at $i$ and ends at $j$ for all $i,j \in \mathcal{V}$. We say a signed graph is \textit{undirected} if $(i,j) \in \mathcal
{E} \iff (j,i) \in \mathcal{E}$ and $s(i,j) = s(j,i)$ for all $(i,j) \in \mathcal{E}$, and it is \textit{directed} otherwise.  For a directed signed graph, a pair of edges sharing the same vertices $(i,j), (j,i) \in \mathcal{E}$ is a \textit{digon}. When $s(i,j) s(j,i) = 1$ for any digon in $\mathcal{E}$ the graph is \textit{digon sign-symmetric}, which means that edges between any two nodes share sign. An undirected signed graph is therefore a digon sign-symmetric graph with the property $A = A^T$.

Given a signed graph $\mathcal{G}$, its underlying unsigned graph $|\mathcal{G}|$ with adjacency matrix $|A|$ is obtained by setting $s(i,j) = 1$ for all $(i,j) \in \mathcal{E}$. A \textit{simple path} from $i$ to $j$ is a sequence of edges starting at vertex $i$ and ending at vertex $j$ that does not visit the same vertex twice, e.g. $(i,l), (l,m), \dots , (k,n), (n,j)$. The \textit{signature of a path} is the product of the signatures of all edges along the path, e.g. $s(i,l)s(l,m)\dots s(k,n) s(n,j)$. A \textit{cycle} is a simple path except that it starts and ends at the same vertex. We say $\mathcal{G}$ is \textit{strongly connected} if there exists a simple directed path in $\mathcal{E}$ from any vertex to any other vertex in $\mathcal{V}$. Equivalently, it is strongly connected if $|A|$ is irreducible. A signed graph $\mathcal{G}$ is \textit{structurally balanced} if it admits a bipartition of nodes $\mathcal{V}_1, \mathcal{V}_2$ such that $\mathcal{V}_1 \cup \mathcal{V}_2 = \mathcal{V}$ and $\mathcal{V}_1 \cap \mathcal{V}_2 = \emptyset$, with $s(i,j) = 1$ for all $(i,j) \in \mathcal{E}$ satisfying $i,j \in \mathcal{V}_q$, $q \in \{1,2\}$, and with $s(i,j) = -1$ for all $(i,j) \in \mathcal{E}$ satisfying $i \in \mathcal{V}_q$, $j \in \mathcal{V}_r$, $q,r \in \{1,2\}$, $q \neq r$. 

We can define a \textit{switching function} $\theta: \mathcal{V} \to \{-1,1\}$ for a given signed graph $\mathcal{G} = (\mathcal{V},\mathcal{E},s)$. This function partitions the vertex set $\mathcal{V}$ into the sets $\mathcal{V}^+$ and $\mathcal{V}^-$ such that $\theta(i) = 1 (-1)$ for all $i \in \mathcal{V}^+ (\mathcal{V}^-)$. We will say that all vertices in $\mathcal{V}^-$ are in the \textit{switching set}, or are being \textit{switched}. The switching function defines a \textit{switching transformation} that maps $\mathcal{G}$ to a new graph $\mathcal{G}' = (\mathcal{V},\mathcal{E}, s')$ that differs only in its signature from the original graph. The switched signature function $s':\mathcal{E} \to \{-1,1\}$ is generated through the relationship $s'(i,j) = \theta(i) s(i,j) \theta(j)$.
Switching a single vertex flips the signature of all edges in $\mathcal{E}$ that point into and out of that vertex. When $\mathcal{V}^-$ contains more than one vertex, a switching transformation flips the signature of all edges that connect vertices in $\mathcal{V}^-$ with those in $\mathcal{V}^+$, while the signature of edges between two vertices in $\mathcal{V}^+$ or in $\mathcal{V}^-$ remains unaltered. Whenever $\mathcal{G}$ and $\mathcal{G}'$ are related by a switching transformation, they are called \textit{switching equivalent}~\cite{zaslavsky1982signed,zaslavsky2013matrices}, or {\it gauge equivalent}~\cite{altafini2012consensus}.

A \textit{switching matrix} associated with a switching transformation is a diagonal matrix $\Theta  =\Theta^{-1} = \operatorname{diag}(\theta(1), \theta(2), \dots, \theta(n))$. The adjacency matrices of $\mathcal{G}$ and $\mathcal{G}'$ are similar, with $ A' = \Theta A \Theta$.
Two signed graphs $\mathcal{G}$, $\mathcal{G}'$ on $n$ vertices are \textit{isomorphic} if there exists an isomorphism between $| \mathcal{G}|$ and $|\mathcal{G}'|$ that preserves the signature of edges in $\mathcal{G}$ and $\mathcal{G}'$.  Formally, $\mathcal{G}$ and $\mathcal{G}'$ are isomorphic if there exists a permutation matrix $P_n$ such that $A' = P_n A P_n^T$. 
If $\mathcal{G}$ is switching equivalent to a graph that is isomorphic to $\mathcal{G}'$, we say the two graphs are \textit{switching isomorphic}. For two switching isomorphic graphs, the adjacency matrices are related by a similarity transformation generated by a switching matrix $\Theta$ and a permutation matrix $P_n$, with
\begin{equation} \label{eq:sim_isomorphism}
A' = P_n \Theta A \Theta P_n^T = \Theta P_n A P_n^T \Theta.
\end{equation}

A graph $\mathcal{G}$ is called \textit{bipartite} if it admits a bipartition of nodes $\mathcal{V}_1, \mathcal{V}_2$ such that $\mathcal{V}_1 \cup \mathcal{V}_2 = \mathcal{V}$ and $\mathcal{V}_1 \edits{\cap} \mathcal{V}_2 = \emptyset$, and for all $(i,j) \in \mathcal{E}$, $i \in \mathcal{V}_q$, $j \in \mathcal{V}_p$, $q,p \in \{1,2\}$ and $q \neq p$. A signed graph $\mathcal{G} = (\mathcal{V}, \mathcal{E}, s)$ is said to be \textit{sign-symmetric} if it is switching isomorphic to its negation $- \mathcal{G} = (\mathcal{V}, \mathcal{E}, -s)$. A bipartite signed graph is trivially sign-symmetric. \edits{See \cite{ghorbani2020sign} for} non-bipartite 
sign-symmetric undirected graphs. 

{

Next, we define two important classes of signed graphs: 

\begin{definition} \label{def:graph-classes}
    Consider a signed graph $\mathcal{G}$ on $n$ vertices. 
    
    \textit{Class I graph.}  $\mathcal{G}$ is switching isomorphic to $\mathcal{G}'$ with an eventually positive adjacency matrix $A'$. 
    
    \textit{Class II graph.} $\mathcal{G}$ is digon-symmetric, strongly connected, and structurally balanced, i.e. switching isomorphic to a $\mathcal{G}'$ with adjacency matrix $A' \succeq 0_{n \times n}$
\end{definition}

\edits{Graph classes I and II are not mutually exclusive.}  We illustrate the distinction between the two in the next example.

\begin{example}
    Consider graphs $\mathcal{G}_1$, $\mathcal{G}_2$ on 4 vertices with adjacency matrices 
    \begin{equation*} \small
        A_1 = \begin{pmatrix}
            0 &  -1 &  1 &  1 \\ 
        1 &  0 & -1 &  1 \\ 
        1 &  1 &  0 & 1\\
        1 & 1 &  1 &  0
        \end{pmatrix}, \ \ A_2= \begin{pmatrix}
        0 & 1 & 0 & 1 \\
        1 & 0 & 1 & 0 \\
        0 & 1 & 0 & 1 \\
        1 & 0 & 1 & 0
        \end{pmatrix}.
    \end{equation*}
    $A_1$ is eventually positive, as $(A_1)^6$ has all positive entries. However, it is not digon symmetric since $(A_1)_{12} = -1$ and $(A_1)_{21} = 1$. Thus, $\mathcal{G}_1$ and any graph $\mathcal{G}_1'$ that is switching isomorphic to $\mathcal{G}_1$ belongs to Class I of Definition \ref{def:graph-classes} but not to Class II. Meanwhile $A_2$ is not eventually positive, but it is digon-symmetric, irreducible, and structurally balanced. Thus, $\mathcal{G}_2$ and any graph $\mathcal{G}_2'$ that is switching isomorphic to $\mathcal{G}_2$ belongs to Class II but not to Class I.
    
\end{example}

In later sections of the \edits{paper}, signed graphs whose adjacency matrices have a simple dominant eigenvalue will be of particular interest in the belief formation problem. In the following Lemma we show that belonging to Class I or Class II of Definition \ref{def:graph-classes} is a sufficient condition for this.

\begin{lemma}[Graphs with simple dominant eigenvalue] \label{lem:dominant-eig}
Consider a signed graph $\mathcal{G}$ on $n$ vertices that is  Class I or Class II or both. 
Then the following statements are true: 

i) The adjacency matrix $A$ has a unique dominant eigenvalue satisfying $\lambda = \rho(A) > \operatorname{Re}(\lambda_i)$ for all $\lambda_i \neq \lambda$ in $\sigma(A)$. If $\mathcal{G}$ is Class I, then additionally $|\lambda| > |\lambda_i|$ for all $\lambda_i \neq \lambda$ in $\sigma(A)$; 

ii) 
There exist a positive vector $\mathbf{v}' \succ \mathbf{0}$, a switching matrix $\Theta$, and a permutation matrix $P_n$ such that $\mathbf{v} = \Theta P_n \mathbf{v}'$ is an eigenvector of $A$ corresponding to its dominant eigenvalue $\lambda$.
\end{lemma}


Strongly connected graphs with all-positive edge signatures, strongly connected structurally balanced graphs, and graphs with an eventually positive adjacency matrix are all special cases of graphs satisfying Lemma \ref{lem:dominant-eig}, \edits{which} belong to Class I and/or to Class II. These graph properties are commonly linked to important features of opinion formation, both for linear and nonlinear models \cite{altafini2012consensus,altafini2014predictable,jiang2016sign,fontan2021role,bizyaeva2022switching,bizyaeva2022sustained,wang2020biased,Wang2022}.  

\section{Nonlinear belief formation dynamics \label{sec:model}}

\subsection{Model statement}

Agents may form beliefs about options in order to choose among them, for example, when mobile agents choose among alternative heading directions, resource-collecting agents choose among alternative patches, and voters choose among candidates in an election. Agents may also form beliefs about options that present as topics, for example, how strongly they support or reject each of a set of topics such as policy protocols or political issues. 
In any of these and other scenarios, network relationships may be important: the beliefs among the agents may be interdependent and the beliefs on the different options (or topics) may be logically interdependent.

We study a network of $N_{\edits{\rm a}}$ agents forming beliefs on $N_{\edits{\rm o}}$ options according to a nonlinear update rule introduced in \cite{bizyaeva2022, franci2022breaking,bizyaeva2022sustained}.
Let the belief state of agent $i$ be represented by the vector $\Zz_i = (z_{i1},z_{i2}, \dots, z_{i N_{\edits{\rm o}}})\edits{^T}  \in \R{N_{\edits{\rm o}}}$, where $z_{ij} \in \R{}$ is the \textit{belief} agent $i$ has about option $j$. In~\cite{franci2022breaking}, $z_{ij}$ is also interpreted as the \textit{value} agent $i$ assigns to option $j$.
Let $\Zz^\dag_j=(z_{1j},\ldots,z_{N_{\edits{\rm a}}j})\edits{^T} \in\R{N_{\edits{\rm a}}}$ represent the belief state of the network on option $j$. 
We say agent $i$ is \textit{neutral} when $\Zz_i = \mathbf{0}$ and \textit{opinionated} otherwise. 
\edits{We model} the strength of belief of agent $i$ about option $j$ as proportional to $|z_{ij}|$, with the agent \textit{favoring} option $j$ when $z_{ij} > 0$ and \textit{rejecting} option $j$ when $z_{ij} < 0$. 
The \textit{network state} 
represents the belief states of all agents, $\Zz = (\Zz_1\edits{^T} , \Zz_2\edits{^T} , \dots, \Zz_{N_{\edits{\rm a}}}\edits{^T} )\edits{^T} \in \R{N}$, 
$N = N_{\edits{\rm a}} N_{\edits{\rm o}}$.

There are two signed directed graphs that are fundamental to the belief formation process \edits{in our model}. 
The first is the \textit{communication graph} among agents, $\mathcal{G}_{\edits{\rm a}} = (\mathcal{V}_{a}, \mathcal{E}_{a},s_{\edits{\rm a}})$ where $\mathcal{V}_{\edits{\rm a}} = \{1, \dots, N_{\edits{\rm a}}\}$ is the vertex set corresponding to the $N_{\edits{\rm a}}$ agents, $\mathcal{E}_{\edits{\rm a}}$ is the edge set, and $s_{a}: \mathcal{E}_{\edits{\rm a}} \to \{-1,1\}$ is the signature of the communication graph $\mathcal{G}_{\edits{\rm a}}$. 
When $e_{ik} \in \mathcal{E}_{\edits{\rm a}}$, agent $k$ is a \textit{neighbor} of agent $i$ and the belief state of agent $k$ influences the belief formation of agent $i$. 
When $s_{\edits{\rm a}}(e_{ik}) = 1$, agent $i$ is \textit{cooperative} towards agent $k$, and whenever $s_{\edits{\rm a}}(e_{ik}) = -1$ it is \textit{competitive} or \textit{antagonistic} towards agent $k$. We assume that $\mathcal{G}_{\edits{\rm a}}$ is \textit{simple}, i.e. contains no self-loops $e_{ii} \not\in \mathcal{E}_{\edits{\rm a}}$ for all $i \in \mathcal{V}_{\edits{\rm a}}$, and there is at most one edge $e_{ik}$ in $\mathcal{E}_{\edits{\rm a}}$ that begins at vertex $i$ and ends at vertex $k$ for all $i,k \in \mathcal{V}_{\edits{\rm a}}$. 
The signed \textit{adjacency matrix} of the communication graph is the matrix $A_{\edits{\rm a}} \in \R{N_{\edits{\rm a}} \times N_{\edits{\rm a}}}$ whose entries are defined as $(A_{\edits{\rm a}})_{ik} = 0$ if $e_{ik} \not \in \mathcal{E}_{\edits{\rm a}}$ and $s_{\edits{\rm a}}(e_{ik})$ otherwise.

\begin{figure}
    \centering
    \includegraphics[width=0.9\columnwidth]{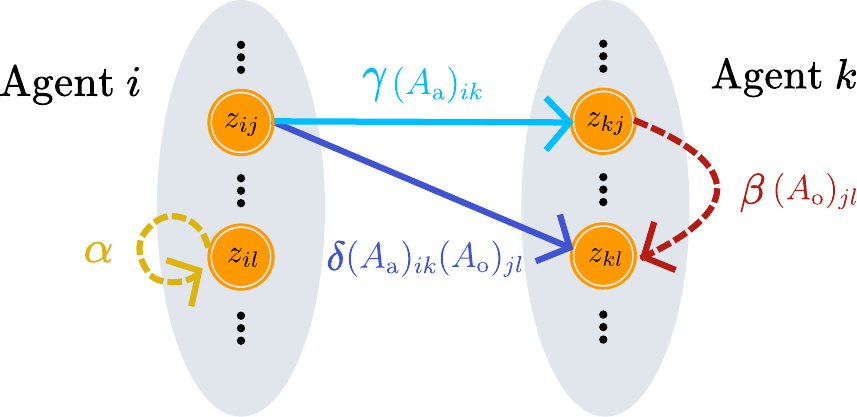}
    \caption{Four types of communication arrows associated to four different cognitive effects on belief formation in model \eqref{eq:dynamics_with_observation}. Arrow direction follows sensing convention.}
    \label{fig:comm-gains}
\end{figure}

The second fundamental graph \edits{in our model} is the \textit{belief system graph} $\mathcal{G}_{\edits{\rm o}} = (\mathcal{V}_{\edits{\rm o}}, \mathcal{E}_{\edits{\rm o}}, s_{\edits{\rm o}})$ that encodes the logical, psychological, or social constraints on the relationships between beliefs on different options \cite{converse2006nature}.  The nodes in the vertex set $\mathcal{V}_{\edits{\rm o}} = \{1, \dots, N_{\edits{\rm o}}\}$ correspond to distinct options (topics).  An edge $e_{jl} \in \mathcal{E}_{\edits{\rm o}}$ thereby signifies \edits{in our model} that formation of beliefs about option $j$ is affected by the beliefs about option $l$. The signature function $\sigma_{\edits{\rm o}}: \mathcal{E}_{\edits{\rm o}} \mapsto \{-1, 1\}$ describes whether two options are positively or negatively aligned according to the belief system \edits{in our model}. We also assume that $\mathcal{G}_{\edits{\rm o}}$ is simple, and we define the signed adjacency matrix $A_{\edits{\rm o}} \in \R{N_{\edits{\rm o}} \times N_{\edits{\rm o}}}$ whose entries are defined as $(A_{\edits{\rm o}})_{jl} = 0$ if $e_{jl} \not \in \mathcal{E}_{\edits{\rm o}}$ and $s_{\edits{\rm a}}(e_{jl})$ otherwise. We assume that all agents share a single global belief system. \edits{This assumption is meaningful from several perspectives. 
A shared belief system may arise when there are logic constraints between different dimensions of belief~\cite{friedkin2016network} or from the existence of shared social axioms, defined in \cite{leung2008psychological} to be ``generalized beliefs...encoded in the form of an assertion about two entities or concepts''. Such axioms may come from a shared political, cultural, organizational, or religious background of individuals in a group and shape collective behaviors in social systems~\cite{leung2008psychological}. Additionally, imposing a shared belief system is convenient to design collective behaviors in groups of artificial agents by reducing the number of design parameters.}

\edits{Let each agent $i$ form a vector $\mathbf{y}_i = (\mathbf{y}_{i k_1}, \dots, \mathbf{y}_{i k_{m_i}} )\edits{^T} \in \R{m_i N_{\edits{\rm o}}}$ of local observations of its $m_i$ neighbors' belief states, where 
$\mathbf{y}_{ik_l} 
= f_i(\Zz_{k_l})$, $l=1,\ldots,m_i$, with $f_i: \R{N_{\edits{\rm o}}} \to \R{N_{\edits{\rm o}}}$.} This formulation allows for noisy, partial, or otherwise imperfect information over the communication network. Agents' beliefs \edits{change} in continuous time according to
\begin{gather} 
   \dot{z}_{ij}\!=\!-d_i \, z_{ij} \!+\!b_{ij} \!+\! 
	u_i\left(  S_{1}\left( \alpha_i \, z_{ij} \!+\! \gamma_i \,  \textstyle\sum_{\substack{k=1 \\ k \neq i}}^\Na(A_{\edits{\rm a}})_{ik} \edits{(\mathbf{y}_{ik})_j}\right) \right.  \nonumber  \\
	\left. + \textstyle\sum_{\substack{l\neq j\\l=1}}^\No S_{2}\!\left( \beta_i \, (A_{\edits{\rm o}})_{jl} z_{il} + \delta_i \, (A_{\edits{\rm o}})_{jl}\textstyle\sum_{\substack{k = 1\\k \neq i}}^\Na (A_{\edits{\rm a}})_{ik} \edits{(\mathbf{y}_{ik})_l}\right)\right). \label{eq:dynamics_with_observation}  
\end{gather}
where $d_i>0$ is a damping coefficient that quantifies the agents' resistance towards \edits{becoming opinionated}, $b_{ij} \in \R{}$ is a $\textit{bias}$ of agent $i$ towards option $j$ or an \textit{input} to agent $i$ about the quality of option $j$, and $\alpha_i, \gamma_i, \beta_i, \delta_i \geq 0$ are the \edits{constant} gains that regulate the relative strengths of four cognitive effects within individual $i$, each corresponding to a distinct arrow type in Fig. \ref{fig:comm-gains}. The gain $\alpha_i$ reflects self-reinforcement or self-appraisal of beliefs within individual $i$.  The gains $\gamma_i, \beta_i, \delta_i$  reflect the relative balance of three distinct sources of cognitive dissonance within agent $i$. We explain these terms in more detail in the following section.

\edits{In our model} $\mathbf{b}_i = (b_{i1}, \dots, b_{i N_{\edits{\rm o}}})\edits{^T} $ is agent $i$'s bias vector and $\mathbf{b} = (\mathbf{b}_1\edits{^T} , \dots, \mathbf{b}_{N_{\edits{\rm a}}}\edits{^T} )\edits{^T} $ the network bias vector. The parameter $u_i \geq 0$ captures \edits{in our model} the \textit{attention} agent $i$ allocates to social interactions
or the \textit{urgency} it feels towards \edits{becoming opinionated}. 
The functions $S_m:\R{}\to[-k_{m1},k_{m2}]$, with $k_{m1},k_{m2} > 0$,  ${S}_m'(y) >0$, $m = 1,2$, saturate social information on each of the options, bounding the net effect of social influence of each topic on the formation of beliefs. We assume that  $S_m(y) = \hat{S}_m(y) + g_m(y)$ with $S_m(0) = 0$, $\hat{S}_m(-y) = - \hat{S}_m(y)$, $\hat{S}_m'(y) >0$, $\hat{S}_m'(0) = 1$,  $g_m(0) = g_m'(0) = 0$, and $g_m(-y) \neq -g_m(y)$, $m = 1,2$.  When $g_1(y) = g_2(y) = 0$, for all $y$, the functions $S_1,S_2$ are odd symmetric and agents are internally equally likely to favor or reject each option. When $g_m(y)\neq 0$, for some $y$, for $m=1$ and/or 2, it perturbs the odd symmetry and captures internal asymmetries in the formation of beliefs. The two saturations $S_1$,$S_2$ are distinct in \eqref{EQ:value_dynamics} to allow for potentially different bounds on the effects of social imitation (influence of beliefs along the same option dimension, $S_1$) and of the ideology (influence of beliefs across option dimensions captured, $S_2$). Alternative formulations are presented in Appendix.

For our analysis, we assume perfect observation by letting $\mathbf{y}_{ik} = f_i(\Zz_k) = \Zz_k$ for all neighboring pairs $i,k \in \mathcal{V}_{\edits{\rm a}}$. We also assume agents to be homogeneous by letting $d_i = d > 0$, $\alpha_i = \alpha \geq 0$, $\gamma_i = \gamma \geq 0$, $\beta_i = \beta \geq 0$, and $\delta_i = \delta \geq 0$ for all $i \in \mathcal{V}_{\edits{\rm a}}$. Note, however, that agents can be affected by heterogeneous biases or inputs $b_{ij}$ \edits{and their network interactions are heterogeneous}. 
Under these assumptions, model \eqref{eq:dynamics_with_observation} becomes
\begin{gather}
   \dot{z}_{ij}=-d \, z_{ij} + 
	u \left(  S_{1}\left( \alpha \, z_{ij} + \gamma \, \textstyle\sum_{\substack{k=1 \\ k \neq i}}^\Na(A_{\edits{\rm a}})_{ik} z_{kj}\right) \right.  \label{EQ:value_dynamics} \\
	\left. + \textstyle\sum_{\substack{l\neq j\\l=1}}^\No S_{2}\left( \beta \, (A_{\edits{\rm o}})_{jl} z_{il} + \delta \, (A_{\edits{\rm o}})_{jl}\textstyle\sum_{\substack{k = 1\\k \neq i}}^\Na (A_{\edits{\rm a}})_{ik} z_{kl}\right)\right)+b_{ij} \nonumber  
\end{gather}
or, equivalently in vector-matrix form,
\begin{multline}
    \dot{\Zz} = - d \Zz + u \mathbf{S}_1\big( ((\alpha \mathcal{I}_{N_{\edits{\rm a}}}  + \gamma A_{\edits{\rm a}} )\otimes \mathcal{I}_{N_{\edits{\rm o}}} )\Zz \big) \\
    + \textstyle\sum_{l = 1}^{N_{\edits{\rm o}}} \mathbf{S}_2\big( ((\beta \mathcal{I}_{N_{\edits{\rm a}}}+ \edits{\delta} A_{\edits{\rm a}}) \otimes M_{l}) \Zz \big) + \mathbf{b} \label{EQ:value_dynamics_vector}
\end{multline}
where $\mathbf{S}_{m}(\mathbf{y)} = (S_m(y_1),\dots, S_m(y_n))\edits{^T}$ for $\mathbf{y} \in \R{n}$, and $M_l \in \R{N_{\edits{\rm o}} \times N_{\edits{\rm o}}}$ is the matrix \edits{with all zeros entries except for column $l$ which is equal to} column $l$ of $A_{\edits{\rm o}}$. 

\edits{Models \eqref{eq:dynamics_with_observation} and \eqref{EQ:value_dynamics}} generalize the opinion dynamics and value formation models of \cite{bizyaeva2022} and \cite{franci2022breaking}. First, in both of the preceding works, the explicit notion of a belief system is absent and there is all-to-all, same-sign coupling between the options. In contrast, a structured belief system is central to the analysis we present here. Second, motivated by an interpretation of opinions on the simplex, \cite{bizyaeva2022} enforces the constraint, not enforced here, that agents' internal beliefs must sum to zero. Lastly, the social networks analyzed in \cite{bizyaeva2022} and \cite{franci2022breaking} were constrained to purely cooperative or purely antagonistic interactions between pairs of agents, whereas here we consider networks with potentially mixed social relationships. 

\edits{Our analysis focuses on the  model \eqref{EQ:value_dynamics}. We summarize the three simplifying assumptions underlying this model:\\ 
1) Agents subscribe to a global belief system, captured by $\mathcal{G}_{\rm o}$;\\
2) Agents have perfect observation of neighbors' beliefs over a social network, captured by $\mathcal{G}_{\rm a}$;\\
3) Agents are homogeneous in their experience of cognitive dissonance, i.e. have identical gains $d,\alpha,\beta,\gamma,\delta$.\\
}
\edits{These simplifying assumptions make our bifurcation analyses tractable. However, the bifurcation results  hold for the more general case with heterogeneous parameters in \eqref{eq:dynamics_with_observation}. Moreover, the bifurcation theory is robust in the sense that the results are robust to small uncertainties and noise, including small imperfections in observations as in  \eqref{eq:dynamics_with_observation}.} 

\subsection{Connection with Networks of Beliefs theory}

The belief formation model \eqref{eq:dynamics_with_observation} has a compelling interpretation in the Networks of Beliefs (NB) theory recently proposed in social psychology \cite{dalege2023networks} and built on three premises: 
\begin{enumerate}
    \item Beliefs can be represented as nodes in a network, with an \textit{internal} network that describes interactions of beliefs within an individual, and an \textit{external} network that describes interactions between individuals;
    \item Belief formation seeks to reduce \textit{dissonance} between one's internal beliefs and beliefs on the social network;
    \item \edits{Several distinct types of} dissonance contribute to the formation of beliefs; \edits{in our model we consider} \textit{personal}, \textit{social}, and \textit{external \edits{ideological}} dissonances. 
    Individuals can allocate different amounts of attention to each. 
\end{enumerate}
In our model, the belief system graph $\mathcal{G}_{\edits{\rm o}}$ describes the internal network and the graph product $\mathcal{G}_{\edits{\rm a}} \times \mathcal{G}_{\edits{\rm o}}$ with adjacency matrix $A_{\edits{\rm a}} \otimes A_{\edits{\rm o}}$ describes the external network in the sense of \cite{dalege2023networks}. NB theory also distinguishes between individuals' \textit{personal} beliefs about a topic and individuals' \textit{social} beliefs, i.e. their perceptions about the beliefs of their neighbors. In \eqref{eq:dynamics_with_observation}, the belief state $z_{ij}$ corresponds to the personal belief of agent $i$ on option $j$ and the vector of local observations $\mathbf{y}_{ik}$ is the set of social beliefs of agent $i$ about agent $k$.

The update of an agent's internal belief $z_{ij}$ in \eqref{eq:dynamics_with_observation} is proportional to  $- \big(z_{ij} -  \frac{u_i}{d_i} F_{ij}(\Zz_i, \mathbf{y}_i) \big)$ where $F_{ij}$ is shorthand for the nonlinear terms grouped by parentheses in \eqref{eq:dynamics_with_observation}. In the terminology of NB theory, the difference in the parentheses represents the magnitude of dissonance individual $i$ experiences due to inconsistencies between its personal belief on option $j$ and its other personal and social beliefs. For unbiased agent $i$, a belief $z_{ij}$  decreases whenever the dissonance term is positive, and increases whenever it is negative, thus minimizing dissonance until equilibrium is reached. 
For an unbiased agent, dissonance is exactly zero at any equilibrium of the model \eqref{eq:dynamics_with_observation}. 

The parameters $\beta_i, \gamma_i, \delta_i$ in \eqref{eq:dynamics_with_observation} are identified with effects of three distinct sources of dissonance within an individual:
\begin{itemize}
    \item $\beta_i$ is identified with  \textit{personal dissonance}, i.e., the inconsistency between an individual's personal belief $z_{ij}$ and all of its other personal beliefs, evaluated according to some logical, moral, or ideological compass encoded in the belief system graph $\mathcal{G}_{\edits{\rm o}}$;
    \item $\gamma_i$ is identified with \textit{social dissonance}, i.e., the inconsistency between an agent's personal belief $z_{ij}$ and its social beliefs $({\bf y}_{ik})_j$ (its perceptions about the beliefs of its neighbors on option $j$); 
    \item $\delta_i$ is identified with \textit{external \edits{ideological } dissonance}, i.e., the inconsistency between an agent's personal belief on option $z_{ij}$ and its social beliefs $({\bf y}_{ik})_l$ (its perceptions about the beliefs of its neighbors on all other options), evaluated according to some logical, moral, or ideological compass encoded in the belief system graph $\mathcal{G}_{\edits{\rm o}}$.
\end{itemize}
The magnitudes of $\beta_i, \gamma_i, \delta_i$ modulate the relative amount of attention allocated by the agents towards each of these effects. 

For an example of these effects, consider forming beliefs on two political topics: increasing the taxation rate for the wealthy and raising immigration rates. According to the conventional left-right political ideology spectrum, a set of beliefs that either supports or opposes both of these issues simultaneously is considered logically consistent. With this in mind, an individual who supports increased taxes but opposes raising immigration will experience \textit{personal} dissonance. Someone who supports increased taxes but perceives that their friends oppose increased taxation will experience \textit{social} dissonance. Finally, an individual who supports increased taxes but perceives that their friends oppose raising immigration will experience \textit{external \edits{ideological}} dissonance when updating their belief on taxation. 
  
NB theory is presented in \cite{dalege2023networks} alongside a computational model of social belief formation, which we refer to as the NB model. Our model is both distinct from and complementary to the NB model of \cite{dalege2023networks}. Firstly, the NB model adopts a statistical physics, i.e., stochastic dynamics, modeling perspective whereas our approach is deterministic. This aspect makes our model amenable to dynamical systems analysis, which we rely on to establish our theoretical results. Secondly, our model incorporates factors not included in the NB model, such as the self-appraisal of personal beliefs, parameterized by $\alpha_i$, and the impact of internal biases $b_{ij}$. Lastly, \edits{external ideological dissonance in our model is defined differently than the \textit{external dissonance} that appears in the NB model}. 
In the NB model, external dissonance arises from inconsistencies between individuals' social beliefs $\mathbf{y}_{ik}$ and the true personal beliefs of their neighbors $\Zz_k$. This definition requires agents to access $\Zz_k$ \textit{in addition} to forming the observation $\mathbf{y}_{ik}$ for each of their neighbors. 
\edits{In the model \eqref{eq:dynamics_with_observation} individuals assess whether their perception of their neighbors' beliefs} appear internally consistent, according to the belief system $\mathcal{G}_{\edits{\rm o}}$. 
Unlike the NB model, in our model agents can experience external dissonance even if their social belief $\mathbf{y}_{ik}$ accurately reflects the true personal belief $\Zz_k$ of their neighbor, \edits{as is the case for the analysis of \eqref{EQ:value_dynamics} we present.} Dissonance minimization as a driving mechanism of belief formation also motivated a recently proposed weighted median model of opinion dynamics \cite{mei2022micro}. Unlike our models \eqref{eq:dynamics_with_observation} \edits{and \eqref{EQ:value_dynamics}}, the weighted median model only considers scalar beliefs and does not account for sources of dissonance beyond social dissonance.

\section{Linear bifurcation analysis \label{sec:lin}}

\subsection{Belief-forming bifurcations}

In this section, we begin our analysis of the homogeneous belief formation model \eqref{EQ:value_dynamics}.The analysis we carry out in this paper draws on standard ideas from the theory of bifurcations in networks and coupled cell systems  \cite{golubitsky2009bifurcations,stewart2011synchrony,stewart2014synchrony,nijholt2019center,golubitsky2023dynamics}. 
We first consider a network of unbiased agents with $b_{ij} = 0$ for all $i \in \mathcal{V}_{\edits{\rm a}}$, $j \in \mathcal{V}_{\edits{\rm o}}$, in which case the neutral state $\Zz = \mathbf{0}$ is an equilibrium of the dynamics for all choices of parameters $d,u,\alpha,\gamma,\beta,\delta$. From~\eqref{EQ:value_dynamics_vector}, it easily follows that the Jacobian matrix at $\Zz = \mathbf{0}$ has the structure
\begin{multline}
    J (\mathbf{0},u) = (- d + u \alpha) \mathcal{I}_{N_{\edits{\rm a}}}\otimes \mathcal{I}_{N_{\edits{\rm o}}} + u \gamma A_{\edits{\rm a}} \otimes \mathcal{I}_{N_{\edits{\rm o}}}\\
    + u \beta \mathcal{I}_{N_{\edits{\rm a}}} \otimes A_{\edits{\rm o}} + u \delta A_{\edits{\rm a}} \otimes A_{\edits{\rm o}}. \label{eq:jac}
\end{multline}
The first term on the left hand side of \eqref{eq:jac} corresponds to the balance of resistance and self-appraisal, the second to social dissonance, the third to personal dissonance, and the fourth to external dissonance.  These terms reflect the four distinct coupling arrows identified in Fig. \ref{fig:comm-gains}. The following lemma shows how spectral properties of $J(\mathbf{0},u)$ depend on the social network and belief system graphs, and model parameters. 

\begin{lemma}[Jacobian spectrum \cite{bizyaeva2022sustained}] \label{prop:eigen}
The following hold for \eqref{eq:jac}. 1) For each $\eta \in \sigma\big(J(\mathbf{0},u)\big)$, there exists $\lambda \in \sigma(A_{\edits{\rm a}})$ and $\mu \in  \sigma(A_{\edits{\rm o}})$ so that 
\begin{equation} \eta = -d + u (\alpha +  \gamma \lambda + \beta \mu + \delta \lambda \mu):= \eta(u,\lambda,\mu); \label{eq:jac_eig}
\end{equation} 2) For each pair $\lambda_i \in \sigma(A_{\edits{\rm a}})$ with an eigenvector $\mathbf{v}_{a,i}$ and $\mu_j \in \sigma(A_{\edits{\rm o}})$ with an eigenvector $\mathbf{v}_{\edits{\rm o},j}$, the vector $\mathbf{v}_{a,i} \otimes \mathbf{v}_{\edits{\rm o},j}$ is an eigenvector of \eqref{eq:jac} with eigenvalue $\eta(u,\lambda_i,\mu_j)$. \label{lem:eigs}
\end{lemma}
\begin{proof}
See \cite[Proposition IV.1]{bizyaeva2022sustained}
\end{proof}
With Lemma \ref{lem:eigs} we have established a bijective correspondence between the Jacobian spectrum $\sigma(J(\mathbf{0}))$ and the product set of the adjacency spectra $\sigma(A_{\edits{\rm a}}),\sigma(A_{\edits{\rm o}})$. For every pair $\lambda \in \sigma(A_{\edits{\rm a}})$, $\mu \in \sigma(A_{\edits{\rm o}})$ we will say $\lambda, \mu$ \textit{generate} the Jacobian eigenvalue $\eta(u,\lambda, \mu)$ defined in \eqref{eq:jac_eig} for some constant choice of model parameters $d,\alpha,\beta,\gamma, \delta$. Furthermore, eigenspaces of $J(\mathbf{0})$ coincide with those of $A_{\edits{\rm a}} \otimes A_{\edits{\rm o}}$, independently of the choice of model parameters. These properties allow us to derive necessary and sufficient conditions for the emergence of a low-dimensional locally attracting invariant manifold in dynamics \eqref{EQ:value_dynamics} as the neutral equilibrium $\Zz = \mathbf{0}$ loses stability. This result generalizes \cite[Theorem IV.1]{bizyaeva2022} to mixed-sign communication and belief networks.

\begin{theorem}[Belief-forming bifurcation] \label{thm:stab} Consider model \eqref{EQ:value_dynamics} with $b_{ij} = 0$ for all $i \in \mathcal{V}_{\edits{\rm a}},j\in\mathcal{V}_{\edits{\rm o}}$ and let 
\begin{align}
\label{eq:Lambda}
    \Lambda = \underset{(\lambda , \mu) \in \sigma(A_{\edits{\rm a}}) \times \sigma(A_{\edits{\rm o}})} {\operatorname{argmax}}\gamma \operatorname{Re}(\lambda)+\beta \operatorname{Re}(\mu) + \delta \operatorname{Re}(\lambda \mu)\,.
\end{align}
Let $(\lambda , \mu)\in\Lambda$ and $\xi_{max} =  \alpha + \gamma \operatorname{Re}(\lambda) + \beta \operatorname{Re}(\mu) + \delta \operatorname{Re}(\lambda \mu)$, and assume $\xi_{max} > 0$.

1) The neutral equilibrium $\Zz = \mathbf{0}$  is locally exponentially stable whenever $0 \leq u < u^*$ and unstable for $u > u^*$, where
\begin{equation}
    u^* = \frac{d}{\alpha + \gamma \operatorname{Re}(\lambda) + \beta \operatorname{Re}(\mu) + \delta \operatorname{Re}(\lambda \mu)} = \frac{d}{\xi_{max}}. \label{eq:ustar}
\end{equation}

2)  Let $|\operatorname{ker}(J(\mathbf{0},u^*))| := k$. There exists a $(k+1)$-dimensional invariant center manifold $W^c \subset \mathbb{R}^{N_{\edits{\rm a}} N_{\edits{\rm o}} + 1}$ passing through $(\Zz,u) = (\mathbf{0},u^*)$, tangent to $\mathcal{N}(J)$ at $u = u^*$.
 All trajectories of \eqref{EQ:value_dynamics} starting at $(\Zz,u)$ near $(\mathbf{0},u^*)$ converge to $W^c$ exponentially as $t \to \infty$. \label{thm:bif}
\end{theorem}
There are several key takeaways from the analysis so far:

1) Unbiased agents that are weakly attentive to social interactions remain neutral and do not form \edits{become opinionated};

2) When network attention exceeds the critical value \eqref{eq:ustar}, non-neutral beliefs emerge as attracting solution sets;

3) The critical attention value $u^*$ is determined by the choice of model parameters and by the eigenvalue pairs of $A_{\edits{\rm a}}$ and $A_{\edits{\rm o}}$ contained in the set $\Lambda$ defined by \eqref{eq:Lambda}. 

A primary aim of this paper is to characterize the properties of solutions of model \eqref{EQ:value_dynamics} that emerge as a result of the bifurcation in which the neutral equilibrium loses stability as $u$ increases above $u^*$. We can already infer the role of the structure of the graphs $\mathcal{G}_{\edits{\rm a}}$, $\mathcal{G}_{\edits{\rm o}}$ in shaping these emerging beliefs. 
When the neutral equilibrium $\Zz$ loses stability, by Theorem \ref{thm:bif} we expect the belief trajectories $\Zz(t)$ to settle on the low-dimensional center manifold $W^c$. Generically, by Lemma \ref{lem:eigs}, the center manifold is given, to first order, by the eigenspace of $A_{\edits{\rm a}} \otimes A_{\edits{\rm o}}$ associated to eigenvalues in the set~\eqref{eq:Lambda}. 
A straightforward eigenvalue problem on $A_{\edits{\rm a}} \otimes A_{\edits{\rm o}}$ thus provides us information both about \textit{where} (i.e., at what parameter values) to expect a critical belief-forming bifurcation in the group, as well as \textit{what pattern} the agents' beliefs will take on post-bifurcation. We expand on these ideas more formally in upcoming sections by studying the onset of belief-forming equilibria and oscillations in model \eqref{EQ:value_dynamics}. 

\subsection{Effect of graph structure and internal and social parameters on belief-forming bifurcations \label{sec:params}} 

In this section we ask the question: how do the communication and belief system graphs jointly with the model parameters select the eigenstructure of $A_{\edits{\rm a}} \otimes A_{\edits{\rm o}}$ (and therefore of $J(\mathbf{0},u)$) that determines the linear bifurcation behavior of \eqref{EQ:value_dynamics}, i.e., which eigenvalues belong to the set $\Lambda$ defined in~\eqref{eq:Lambda}?

Let an eigenvalue pair $(\lambda,\mu)$ be a \textit{principal eigenvalue pair} if it belongs to $\Lambda$. Let an eigenspace of $A_{\edits{\rm a}} \otimes A_{\edits{\rm o}}$ be a \textit{principal eigenspace} if its associated eigenvalue pairs are principal, i.e. it is the leading eigenspace of $J(\mathbf{0},u)$.
Also let  $\lambda_{max} = \operatorname{max}_{\lambda_i \in \sigma(A_{\edits{\rm a}})} \{\operatorname{Re}(\lambda_i)\}$,  $\mu_{max}  = \operatorname{max}_{\mu_j \in \sigma(A_{\edits{\rm o}}) }\{\operatorname{Re}(\mu_j)\}$, and $(\lambda \mu)_{max} = \operatorname{max}_{\lambda_i \in \sigma(A_{\edits{\rm a}}), \mu_j\in \sigma(A_{\edits{\rm o}})}  \{ \operatorname{Re}(\lambda_i \mu_j) \}$. We define two sets of eigenvalues that turn out to be key in determining the principal eigenstructure of $A_{\edits{\rm a}}\otimes A_{\edits{\rm o}}$:

\begin{enumerate}
\item Let $\Lambda_{1}$ be a set of ordered pairs $(\lambda,\mu) \in \sigma(A_{\edits{\rm a}}) \times \sigma(A_{\edits{\rm o}})$ for which $\operatorname{Re}(\lambda) = \lambda_{max}$,  $\operatorname{Re}(\mu)  = \mu_{max}$. 

\item Let $\Lambda_{2}$ be a set of ordered pairs $(\lambda,\mu)\in \sigma(A_{\edits{\rm a}}) \times \sigma(A_{\edits{\rm o}})$ for which $\operatorname{Re}(\lambda \mu) =(\lambda \mu)_{max}$.
\end{enumerate}

An \textit{indecision-breaking bifurcation} in \eqref{EQ:value_dynamics} typically happens either along the product of leading eigenspaces of $A_{\edits{\rm a}},A_{\edits{\rm o}}$, or along the leading eigenspace of the graph product $A_{\edits{\rm a}} \otimes A_{\edits{\rm o}}$. The set $\Lambda_1$ is associated with the product of leading eigenspaces of $A_{\edits{\rm a}},A_{\edits{\rm o}}$ and its cardinality reflects the dimension of this eigenspace,  while the set $\Lambda_2$ is similarly associated with the leading eigenspace of $A_{\edits{\rm a}}\otimes A_{\edits{\rm o}}$. Before we formalize this observation, in the following lemma we characterize graphs for which the structure of the sets $\Lambda_1$ and/or $\Lambda_2$ is easily identified. Observe that $\Lambda_1,\Lambda_2$ are not necessarily disjoint and for many common choices of $\mathcal{G}_{\edits{\rm a}}$ and $\mathcal{G}_{\edits{\rm o}}$, $\Lambda_1 \subseteq \Lambda_2$. Let Class I and Class II graphs be defined as in Definition \ref{def:graph-classes}.

\begin{lemma}[Graph conditions]


     i) Suppose $\mathcal{G}_{\edits{\rm a}}$ and $\mathcal{G}_{\edits{\rm o}}$ are each either Class I or Class II; then $\Lambda_1 = \{(\lambda_{max},\mu_{max})\}$;

     ii) Suppose $\mathcal{G}_{\edits{\rm a}}$ and $\mathcal{G}_{\edits{\rm o}}$ are both Class I; then $\Lambda_2 = \Lambda_1 = \{(\lambda_{max},\mu_{max})\}$;

     iii) Suppose $\mathcal{G}_{\edits{\rm a}}$ and $\mathcal{G}_{\edits{\rm o}}$ are both Class II. If $\mathcal{G}_{\edits{\rm a}}$ and $\mathcal{G}_{\edits{\rm o}}$ are both undirected and sign-symmetric, then
     
     $\Lambda_2 =  \{ (\lambda_{max},\mu_{max}),(-\lambda_{max},-\mu_{max}) \}$;

     iv) Suppose $\mathcal{G}_{\edits{\rm a}}$ and $\mathcal{G}_{\edits{\rm o}}$ 
     are both Class II and both undirected. If the characteristic polynomial $\operatorname{det}(xI - A) = \sum_{i \in \mathcal{V}} c_i x^{n-i}$  for both $A_{\edits{\rm a}}, A_{\edits{\rm o}}$ satisfies $c_i = 0$ for any even $i$, the set $\Lambda_2 =  \{ (\lambda_{max},\mu_{max}),(-\lambda_{max},-\mu_{max}) \}$; otherwise, $\Lambda_2 = \Lambda_1 = \{ (\lambda_{max}, \mu_{max}) \}$.\label{lem:graph-cond}
\end{lemma}


We now provide conditions for the principal eigenstructure $\Lambda$ to be easily characterized by $\Lambda_1$ or $\Lambda_2$.


\begin{proposition}
    Consider \eqref{EQ:value_dynamics} with communication graph $\mathcal{G}_{\edits{\rm a}}$ and belief system graph $\mathcal{G}_{\edits{\rm o}}$.
    
    1)  Suppose $\beta > 0, \gamma > 0$. There exists small $\varepsilon$ such that whenever $\delta < \varepsilon$,  $\Lambda = \Lambda_{1}$.
    
    2) Suppose $\delta > 0$. There exist small $\varepsilon_1$, $\varepsilon_2$ such that whenever $\gamma < \varepsilon_1$ and $\beta < \varepsilon_2$,  $\Lambda = \Lambda_2$. \label{prop:param-effects-1}
\end{proposition}

In Proposition \ref{prop:param-effects-1} we have shown that when external dissonance is weak ($\delta$ is zero or small), then $\Lambda=\Lambda_1$, while when social dissonance and personal dissonance are weak ($\gamma$, $\beta$ are zero or small), then $\Lambda=\Lambda_2$. We now prove a similar result but in the converse limit: if social and/or personal dissonance (external dissonance) are (is) sufficiently strong, then $\Lambda=\Lambda_1$ ($\Lambda=\Lambda_2)$.

\begin{proposition}
\label{prop:param-effects-2}

Consider \eqref{EQ:value_dynamics} with $u > 0$ a communication graph $\mathcal{G}_{\edits{\rm a}}$, and a belief system graph $\mathcal{G}_{\edits{\rm o}}$.
Define the second largest real part of eigenvalues of $A_{\edits{\rm a}}$ and $A_{\edits{\rm o}}$ as $\lambda_2 = \max_{\lambda \in \sigma(A_{\edits{\rm a}}) s.t. \operatorname{Re}(\lambda) \neq \lambda_{max}} \{ \operatorname{Re}(\lambda) \}$, $\mu_2 = \max_{\mu \in \sigma(A_{\edits{\rm o}}) s.t. \operatorname{Re}(\mu) \neq \mu_{max} }\{ \operatorname{Re}(\mu) \}$.



1) 
There exists a critical value $K_c > 0$, with all other parameters held constant, such that if $\gamma (\lambda_{max} - \lambda_2) + \beta (\mu_{max} - \mu_2) > K_c$, then $(\lambda,\mu) \in \Lambda$ if and only if $\operatorname{Re}(\lambda) = \lambda_{max}$ and $\operatorname{Re}(\mu) =\mu_{max}$, i.e. $\Lambda = \Lambda_1$; 

2) 
There exists a critical value $\delta_c$, with all other parameters held constant, such that if $\delta > \delta_c$, then $(\lambda,\mu) \in \Lambda$ if and only if $\operatorname{Re}(\lambda \mu) = (\lambda \mu)_{max}$, i.e. $\Lambda = \Lambda_2$.

\end{proposition}

\edits{Propositions~\ref{prop:param-effects-1},~\ref{prop:param-effects-2} hold under the stated assumptions, thusincluding for graphs that are neither Class I or II.}

\subsection{Network of Beliefs theory interpretation \label{sec:NB_interp}}

This section interprets the results of Propositions~\ref{prop:param-effects-1},~\ref{prop:param-effects-2} in the context of the NB theory. A few key takeaways:\\
\textit{1) Self-appraisal tunes the amount of attention or urgency needed \edits{for agents to become opinionated}}. The parameter $\alpha$ is inversely proportional to the critical attention magnitude \eqref{eq:ustar}. When agents are strongly self-appraising, the belief-forming bifurcation occurs at a lower $u^*$.  Because $\alpha$ does not play any role in selecting the principal eigenspace along which belief formation happens, this is the primary effect of self-appraisal.\\
\textit{2) Social dissonance drives network beliefs towards the leading eigenspaces of $\mathcal{G}_{\edits{\rm a}}$ and $\mathcal{G}_{\edits{\rm o}}$.} By Propositions \ref{prop:param-effects-1}, \ref{prop:param-effects-2} when the social dissonance parameter $\gamma$ dominates, the eigenvalue set $\Lambda_1$ is principal for any choice of $\mathcal{G}_{\edits{\rm a}}$ and $\mathcal{G}_{\edits{\rm o}}$. In this case, the principal eigenspace of $J(\mathbf{0},u)$ is a product of the leading eigenspaces of $A_{\edits{\rm a}}$ and $A_{\edits{\rm o}}$. In light of this observation, we call a belief-forming bifurcation for which $\Lambda_1$ is the principal eigenvalue set \textit{social dissonance-driven}. \\
\textit{3) External dissonance reshapes belief formation but only on specific communication and belief system graphs.} By Propositions \ref{prop:param-effects-1}, \ref{prop:param-effects-2}, when the external dissonance parameter $\delta$ is large enough, the eigenvalue set $\Lambda_2$ is principal and the principal eigenspace of $J(\mathbf{0},u)$ coincides with the leading eigenspace of the graph product $A_{\edits{\rm a}} \otimes A_{\edits{\rm o}}$. Hence, whenever $\mathcal{G}_{\edits{\rm a}}$ and $\mathcal{G}_{\edits{\rm o}}$ are such that $\Lambda_2 \neq \Lambda_1$, a sufficiently large external dissonance weight can drastically change the belief-forming behavior for two reasons. First, the manifold along which the belief-forming bifurcation occurs will be different in the external dissonance-driven and the social dissonance-driven case. Second, the \textit{type} of bifurcation that occurs inside the bifurcation manifold will in general be different, e.g. whenever the two sets have different cardinalities. On the same communication and belief system graphs $\mathcal{G}_{\edits{\rm a}}$, $\mathcal{G}_{\edits{\rm o}}$, social dissonance and external dissonance-driven bifurcations in \eqref{EQ:value_dynamics} can lead to sharply different belief-forming behaviors.\\
\textit{4) Personal dissonance amplifies the effects of social dissonance.} Another important observation from Propositions \ref{prop:param-effects-1}, \ref{prop:param-effects-2} is that the effect of personal dissonance parameter $\beta$ is coupled with that of the social dissonance parameter $\gamma$, rather than with that of the external dissonance parameter $\delta$. Even when the social dissonance coupling $\gamma$ is weak, a strong level of personal dissonance $\beta$ will lead to a social dissonance-driven bifurcation shaped by the eigenvalue set $\Lambda_1$. Agents are more susceptible to social pressure when they are less secure internally about their personal identity with respect to the belief system $\mathcal{G}_{\edits{\rm o}}$. 
\begin{figure}
    \centering
    \includegraphics[width=0.8\linewidth]{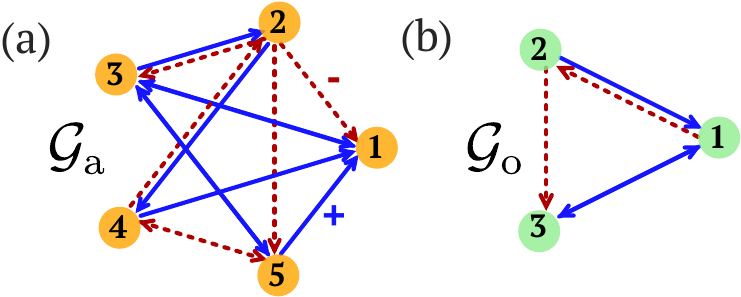}
    \caption{ (a) Communication graph $\mathcal{G}_{\edits{\rm a}}$ for five agents and (b) belief system graph $\mathcal{G}_{\edits{\rm o}}$ for three options with signed adjacency matrices \eqref{eq:adj_mat_ex}. Solid blue edges are positive connections and red dashed edges are negative connections. Arrow direction follows sensing convention.}
    \label{fig:graphs}
\end{figure}
\begin{example} \label{ex:GaGo}
Consider a network of 5 agents evaluating 3 options according to \eqref{EQ:value_dynamics}, with communication graph $\mathcal{G}_{\edits{\rm a}}$ and belief system graph $\mathcal{G}_{\edits{\rm o}}$ pictured in Fig. \ref{fig:graphs}. Then
\begin{equation} \small
    A_{\edits{\rm a}} = \begin{pmatrix}  0 &  0 &  1 &  0 &  0\\
       -1 &  0 & -1&  1& -1 \\
       1&  1&  0&  0&  1\\
       1& -1&  0&  0& -1\\
       1&  0&  1& -1&  0 \end{pmatrix}, \ 
       A_{\edits{\rm o}} = \begin{pmatrix} 0 & -1&  1 \\
        1&  0& -1\\
        1&  0&  0 \end{pmatrix}\! . \label{eq:adj_mat_ex}
\end{equation}
 Rounded to three decimal points, the eigenvalues of $A_{\edits{\rm a}}$ are $ \lambda_1 \approx 0.823, \  \lambda_{2,3} \approx 0.745 \pm 1.106i, \ \lambda_{4,5} \approx -1.157 \pm 0.327i $
and the eigenvalues of $A_{\edits{\rm o}}$ are $ \mu_1 = 1, \ \mu_{2,3} = - \frac{1}{2} \pm \frac{\sqrt{3}}{2}i$.
The set of leading eigenvalues of $A_{\edits{\rm a}},A_{\edits{\rm o}}$ is $\Lambda_1 = \{(\lambda_1,\mu_1)\}$  and the set of eigenvalue pairs that maximize the real part of the eigenvalue products of the two graphs is $\Lambda_2 = \{ (\lambda_{4},\mu_3), (\lambda_5,\mu_2) \}$. For belief formation dynamics \eqref{EQ:value_dynamics} defined on these graphs, this means that a social dissonance-driven bifurcation will happen along a one-dimensional manifold, and an external dissonance-driven bifurcation will happen along a different, two-dimensional manifold.
\end{example}

\section{Multi-stable belief-formation: pitchfork bifurcation \label{sec:pitchfork}} 

In this section we characterize what types of limit set solutions for dynamics \eqref{EQ:value_dynamics} arise at a belief-forming bifurcation. We focus here on the simplest, but also generic, case in which the leading eigenvalue of the Jacobian matrix \eqref{eq:jac} is real and simple. In light of Theorem~\ref{thm:bif}, this amounts to requiring that the set  of principal eigenvalues $\Lambda$ contains exactly one pair of real-valued eigenvalues of $A_{\edits{\rm a}}$ and $A_{\edits{\rm o}}$. 
In the following Theorem we prove that under this assumption, the dynamics \eqref{EQ:value_dynamics} exhibit a pitchfork bifurcation.

\begin{theorem}[Multi-option pitchfork bifurcation] \label{thm:pitchfork-No}
Consider \eqref{EQ:value_dynamics} with communication graph $\mathcal{G}_{\edits{\rm a}}$ and belief system graph $\mathcal{G}_{\edits{\rm o}}$, and suppose $\Lambda = \{ (\lambda_{\edits{\rm a}}, \mu_{\edits{\rm o}}) \}$ with $\lambda_{\edits{\rm a}},\mu_{\edits{\rm o}} \in \R{}$. 
Assume $K_0 := \alpha + \gamma \lambda_{\edits{\rm a}} + \beta \mu_{\edits{\rm o}} + \delta\lambda_{\edits{\rm a}} \mu_{\edits{\rm o}}   > 0$. 
Let $\mathbf{v}_{\edits{\rm a}},\mathbf{w}_{\edits{\rm a}} \in \mathbb{R}^{N_{\edits{\rm a}}}$ and $\mathbf{v}_{\edits{\rm o}},\mathbf{w}_{\edits{\rm o}} \in \mathbb{R}^{N_{\edits{\rm o}}}$ be the right and left eigenvectors of $A_{\edits{\rm a}}$ and $A_{\edits{\rm o}}$ corresponding to $\lambda_{\edits{\rm a}}$ and $\mu_{\edits{\rm o}}$, respectively, normalized to satisfy $\langle \mathbf{w}_{\edits{\rm a}},\mathbf{v}_{\edits{\rm a}}\rangle =1$, $\langle \mathbf{w}_{\edits{\rm o}},\mathbf{v}_{\edits{\rm o}}\rangle =1$. Let
$K_1 = \left(S_1'''(0) (\alpha\! +\!\gamma \lambda_{\edits{\rm a}})^3\! +\! S_2'''(0) (\beta \mu_{\edits{\rm o}} \! +\! \delta \lambda_{\edits{\rm a}}\mu_{\edits{\rm o}}  )^3 \right)\! \langle \mathbf{w}_{\edits{\rm a}}, \mathbf{v}_{\edits{\rm a}}^3 \rangle  \langle \mathbf{w}_{\edits{\rm o}},  \mathbf{v}_{\edits{\rm o}} ^3\rangle$ and $K_2 := \big(g_1''(0) (\alpha +\gamma \lambda_{\edits{\rm a}})^2 +  g_2''(0) (\beta \mu_{\edits{\rm o}}  + \delta \lambda_{\edits{\rm a}}\mu_{\edits{\rm o}})^2 \big)
\langle \mathbf{w}_{\edits{\rm a}}, \mathbf{v}_{\edits{\rm a}}^2 \rangle \langle \mathbf{w}_{\edits{\rm o}},  \mathbf{v}_{\edits{\rm o}} ^2\rangle \neq 0$ where $\mathbf{x}^2 = \mathbf{x} \odot \mathbf{x}$ and $\mathbf{x}^3 = \mathbf{x} \odot \mathbf{x} \odot \mathbf{x}$.

1) \textbf{Unbiased agents with symmetric response.} Suppose $S_1,S_2$ have odd symmetry ($g_1(x) = g_2(x) = 0$). At $(\Zz,u,\mathbf{b})=(\mathbf{0},u^*,\mathbf{0})$, the system undergoes a symmetric pitchfork bifurcation. All bifurcation branches are tangent at $(\Zz,u)=(\mathbf{0},u^*)$ to $\operatorname{span}\{\mathbf{v}_{\edits{\rm a}} \otimes \mathbf{v}_{\edits{\rm o}}\}$. If $ K_1 < 0 (> 0)$, 
the non-trivial equilibria $\Zz^*(u),-\Zz^*(u)\neq \mathbf{0}$ bifurcating from $\mathbf{0}$ exist for $u>u^*$ ($u<u^*$), $\lvert u-u^*\rvert$ sufficiently small, and are locally exponentially stable (unstable).


2) \textbf{Unbiased agents with asymmetric response.} 
Suppose $S_1,S_2$ are asymmetric ($g_1(x) \neq 0$ and/or $g_2(x) \neq 0$), and $|g_1''(0)|$, $|g_2''(0)|$ are sufficiently small. At $(\Zz, u, \mathbf{b}) = (\mathbf{0},u^*,\mathbf{0})$, the system undergoes a transcritical bifurcation. All bifurcation branches are tangent at $(\Zz,u)=(\mathbf{0},u^*)$ to $\operatorname{span}\{\mathbf{v}_{\edits{\rm a}} \otimes \mathbf{v}_{\edits{\rm o}}\}$. The non-trivial equilibria $\Zz_1^*(u)$,$\Zz_2^*(u) \neq \mathbf{0}$ exist for $|u - u^*|$ sufficiently small and satisfy  $\langle \mathbf{v}_{\edits{\rm a}} \otimes \mathbf{v}_{\edits{\rm o}} , \Zz_1^*(u) \rangle > 0$, $\langle \mathbf{v}_{\edits{\rm a}} \otimes \mathbf{v}_{\edits{\rm o}} , \Zz_2^*(u) \rangle < 0$. If $K_2 > 0 \ (< 0)$, $\Zz_1^*(u)$ ($\Zz_2^*(u)$) appears for $u < u^*$ and is unstable, while  $\Zz_2^*(u)$ ($\Zz_1^*(u))$ appears for $u > u^*$ and is locally exponentially stable.



3) \textbf{Biased agents.} Suppose $\mathbf{b} \neq \mathbf{0}$ and $\|\mathbf{b}\|$ is small. Whenever $\langle \mathbf{w}_{\edits{\rm a}} \otimes \mathbf{w}_{\edits{\rm o}},  \mathbf{b} \rangle > 0 \ (< 0)$, for sufficiently small $|u -u^*|$, the local bifurcation diagram of \eqref{EQ:value_dynamics} has a unique, locally exponentially stable equilibrium point $\Zz^*(u)$. This equilibrium satisfies $\langle \mathbf{v}_{\edits{\rm a}} \otimes \mathbf{v}_{\edits{\rm o}},  \Zz^*(u) \rangle > 0 \ (< 0)$.
\end{theorem} 

\edits{Note that the assumption $\Lambda = \{(\lambda_a,\mu_o)\}$ of Theorem \ref{thm:pitchfork-No} implies that $\delta>0$ and/or $\beta,\gamma>0$. Additionally, note that Theorem \ref{thm:pitchfork-No} holds for general signed graphs satisfying the stated assumptions, which includes graphs that are neither Class I or II, and some that are not strongly connected.}

\begin{figure}
    \centering
    \includegraphics[width=0.8\columnwidth]{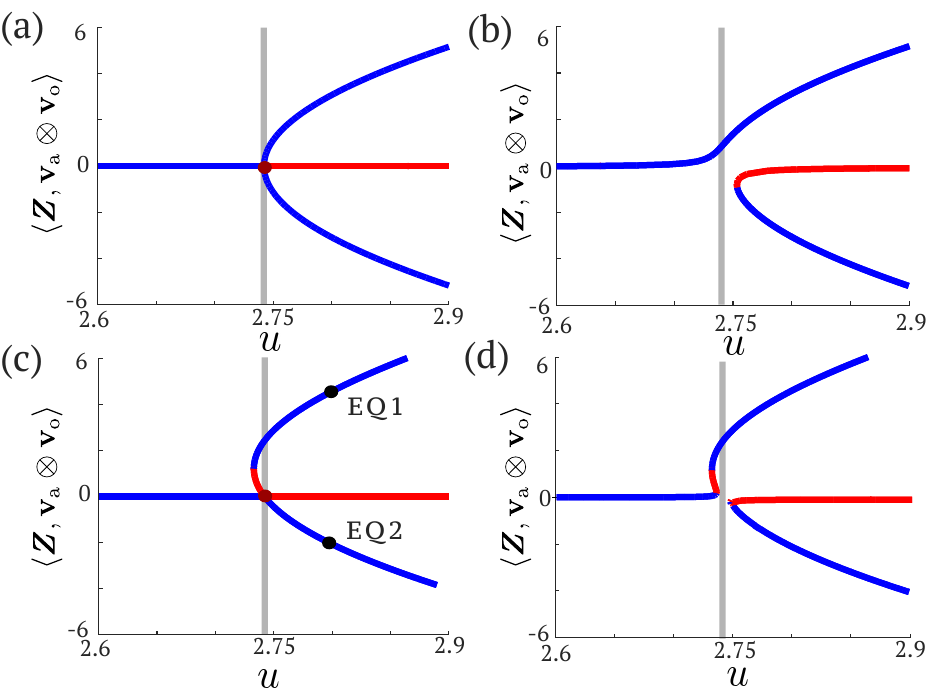}
    \caption{ Bifurcation diagrams for Example \ref{ex:pitchfork}, generated with MatCont numerical continuation package \cite{matcont}. Red (blue) lines represent unstable (stable) equilibria. $ S_1(y) = \tanh(y + \varepsilon_1 \tanh(y^2))$,  $S_2(y) = \frac{1}{2}\tanh(2y + 2\varepsilon_2 \tanh(y^2))
$.   (a) \edits{Symmetric pitchfork bifurcation at $u = u^*$ with} $\varepsilon_1 = \varepsilon_2 = 0$, $\langle \mathbf{w}_{\edits{\rm a}}\otimes \mathbf{w}_{\edits{\rm o}}, \mathbf{b}\rangle = 0$; (b) \edits{unfolding of symmetric pitchfork bifurcation with} $\varepsilon_1 = \varepsilon_2 = 0$, $\langle \mathbf{w}_{\edits{\rm a}}\otimes \mathbf{w}_{\edits{\rm o}}, \mathbf{b}\rangle > 0$; (c)  \edits{unfolding of pitchfork at $u = u^*$ with} $\varepsilon_1 = \varepsilon_2 = 0.1$, $\langle \mathbf{w}_{ \edits{\edits{\rm a}}}\otimes \mathbf{w}_{\edits{\rm o}}, \mathbf{b}\rangle = 0$; (d) \edits{unfolding of pithfork with} $\varepsilon_1 = \varepsilon_2 = 0.1 $, $\langle \mathbf{w}_{\edits{\rm a}}\otimes \mathbf{w}_{\edits{\rm o}}, \mathbf{b}\rangle > 0$. For (b) and (d), $b_{11} = 0.001$, $b_{22} = 0.003$, $b_{53} = -0.002$ and $b_{ij} = 0$ for all other $i\in \mathcal{V}_{\edits{\rm a}}, j\in \mathcal{V}_{\edits{\rm o}}$. Vertical gray line indicates the bifurcation point $u = u^* \approx 2.742$. Other parameters: $d = 1$, $\alpha = \gamma = \beta = \delta = 0.1$.}
    \label{fig:pitchfork}
\end{figure}

\begin{example} \label{ex:pitchfork}
To illustrate Theorem \ref{thm:pitchfork-No} we consider model \eqref{EQ:value_dynamics} with the communication and belief system graphs of Fig. \ref{fig:graphs} and Example \ref{ex:GaGo}. Given model parametrization described in Fig. \ref{fig:pitchfork}, it can easily be established that $\Lambda = \Lambda_1 = \{(\lambda_1, \mu_1)\}$, i.e., the indecision-breaking bifurcation is social dissonance-driven, and the assumptions of Theorem \ref{thm:pitchfork-No} are satisfied. In Fig. \ref{fig:pitchfork} we show four bifurcation diagrams that summarize the predictions of Theorem \ref{thm:pitchfork-No}: a symmetric pitchfork bifurcation for unbiased agents with odd symmetric $S_1$,$S_2$ in (a), an unfolding of a symmetric pitchfork for agents with small biases in (b), a transcritical bifurcation for unbiased agents with asymmetric response in (c), and an unfolding of the asymmetric case for agents with small biases in (d). Observe that in the asymmetric diagram (c), the unstable branch of nontrivial equilibria regains stability in a saddle-node bifurcation and for $u > u^*$ there is a bistability between the two nontrivial branches, similarly to the odd symmetric case. This is expected in general as long as $S_1''(0)$ and $S_2''(0)$ are sufficiently small, as a consequence of unfolding theory for a pitchfork bifurcation - see \cite[Chapter I \S 1]{Golubitsky1985} for a discussion. In this sense, the bistability of belief equilibria is a feature of dynamics \eqref{EQ:value_dynamics} that is robust to asymmetries in the agents' internal and social responses. Fig. \ref{fig:EQ_options} shows simulated belief trajectories settling on EQ1 and EQ2 of Fig. \ref{fig:pitchfork}. At steady state, the vectors $\Zz^\dag_j$, $j=1,\ldots,N_{\edits{\rm o}}$, of agents' beliefs about a given option in (a) and (b) are roughly (modulo higher-order in the pitchfork center manifold expansion) aligned with communication graph eigenvectors \edits{$\mathbf{v}_{\edits{\rm a}} \approx (0.463,  -0.272, 0.381, 0.743,0.123)^T$. }
Similarly, at steady state, the vectors $\Zz_i$, $i=1,\ldots,N_{\edits{\rm o}}$, of an agent's beliefs about the set of option in (c) are roughly aligned with the belief system eigenvector \edits{$\mathbf{v}_{\edits{\rm o}} = (1,0,1)^T$}. 
In particular, note that all beliefs about option 2 remain close to neutral because $(\mathbf{v}_{\edits{\rm o}})_2 = 0$.

\begin{figure}
    \centering
    \includegraphics[width=\columnwidth]{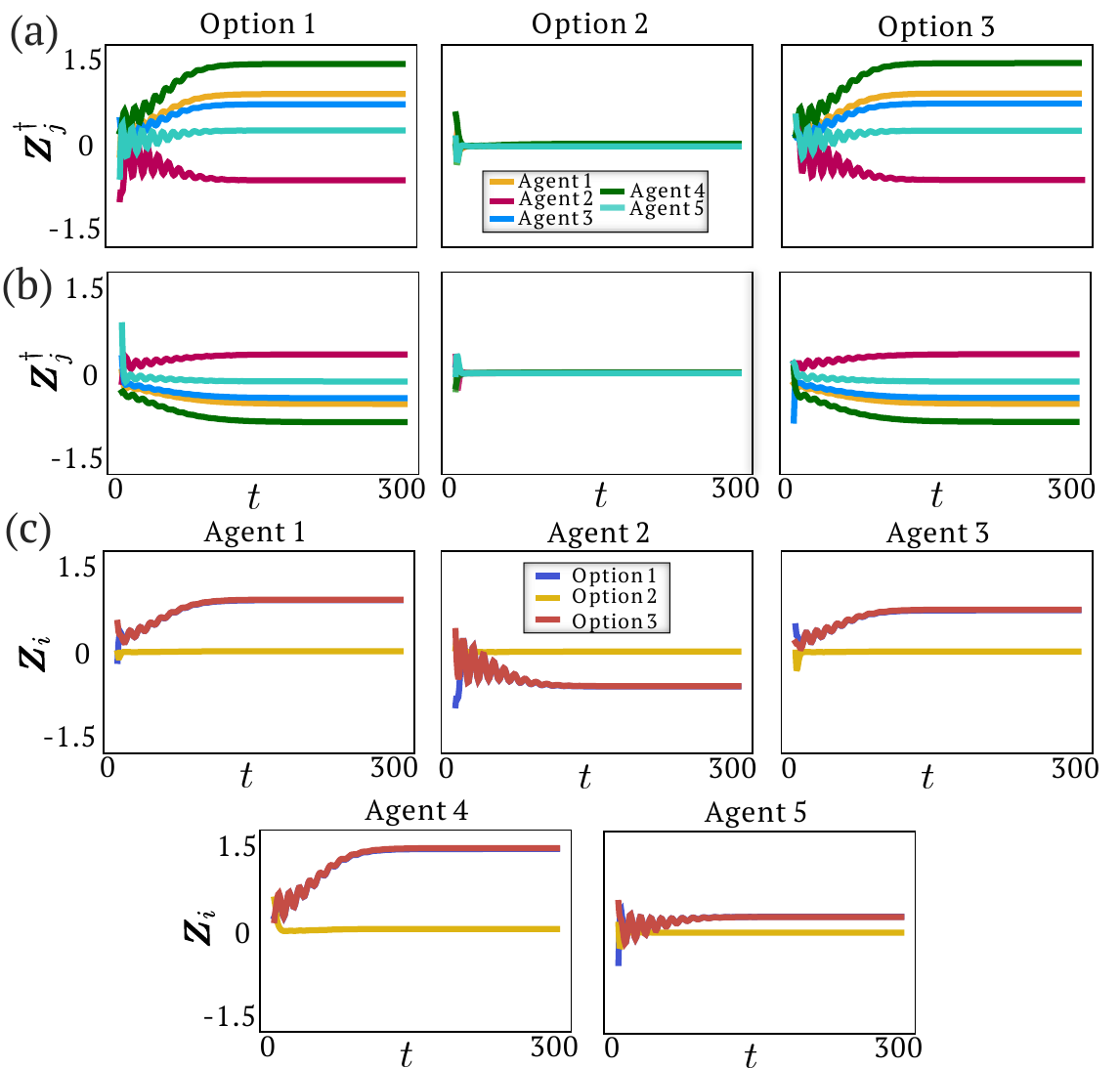}
    \caption{Trajectories $z_{ij}(t)$ for belief formation dynamics \eqref{EQ:value_dynamics} for 5 agents evaluating 3 options on $\mathcal{G}_{\edits{\rm a}}$, $\mathcal{G}_{\edits{\rm o}}$ of Fig. \ref{fig:graphs} from random initial conditions.    The figure illustrates the distribution of agents' states $z_{ij}$ along each option dimension as the network settles to (a) EQ1 from Figure \ref{fig:pitchfork}(c); (b) EQ2 from Figure \ref{fig:pitchfork}(c); (c) shows same trajectories as (a), grouped by agent instead of option. Belief trajectories on options 1 and 3 overlap for most of the simulation for all agents. Parameters: $u = u^* + 0.05 \approx 2.792$, all other parameters as in Fig. \ref{fig:pitchfork}(c).}
    \label{fig:EQ_options}
\end{figure}


\end{example}

Next, we establish sufficient conditions on the structure of communication and belief system graphs for the existence of a pitchfork bifurcation in \eqref{EQ:value_dynamics}. Recall that a social dissonance-driven bifurcation happens whenever the set of bifurcating dominant eigenvalues \eqref{eq:Lambda} consists of the leading eigenvalues of signed adjacency matrices $A_{\edits{\rm a}}$ and $A_{\edits{\rm o}}$. 
Given a graph $\mathcal{G}$ with adjacency matrix $A$ that is either Class I or II (by Definition~\ref{def:graph-classes}), let $\theta: \mathcal{V} \to \{1,-1\}$ be the switching function with switching matrix $\Theta$, and let $P_n$ be a permutation matrix for which $\Theta P_n A P_n^T\Theta$ is eventually positive if $\mathcal{G}$ is Class I or nonnegative if $\mathcal{G}$ is Class II. \edits{Recall that \textit{social dissonance-driven} bifurcations and \textit{external dissonance-driven} bifurcations happen when eigenvalue sets $\Lambda_1$ and $\Lambda_2$, respectively, are dominant (see Sections \ref{sec:params},\ref{sec:NB_interp}).} 

\begin{corollary} \label{cor:pitchfork_conds}

Suppose $A_{\edits{\rm a}}$,$A_{\edits{\rm o}}$ each has a simple leading eigenvalue $\lambda_{max},\mu_{max}$, respectively. Let $\mathbf{v}_{\edits{\rm a}}$,$\mathbf{v}_{\edits{\rm o}}$ be the corresponding eigenvectors. Define the sets $\mathcal{V}_{\edits{\rm a}}^0 = \{i \in \mathcal{V}_{\edits{\rm a}} \ \ s.t. \ \ (\mathbf{v}_{\edits{\rm a}})_i = 0 \} \subset \mathcal{V}_{\edits{\rm a}}$ and $\mathcal{V}_{\edits{\rm o}}^0 = \{j \in \mathcal{V}_{\edits{\rm o}} \ \ s.t. \ \ (\mathbf{v}_{\edits{\rm o}})_j = 0 \} \subset \mathcal{V}_{\edits{\rm o}}$. The following statements hold:

    i) Suppose $\mathcal{G}_{\edits{\rm a}}$ and $\mathcal{G}_{\edits{\rm o}}$ are each either Class I or Class II; then any social dissonance-driven indecision-breaking bifurcation in \eqref{EQ:value_dynamics} must be a pitchfork bifurcation. Let $\Zz^* = (z_{11}^*, z_{12}^*, \dots, z_{N_{\edits{\rm a}} N_{\edits{\rm o}}}^*)$ be an equilibrium on one of the new solution branches of the pitchfork bifurcation with sufficiently small $\lvert u-u^* \rvert $. Then $z_{ij}^* \neq 0$ for all $i\in \mathcal{V}_{\edits{\rm a}}$, $j \in \mathcal{V}_{\edits{\rm o}}$.

    ii) Suppose $\mathcal{G}_{\edits{\rm a}}$ and $\mathcal{G}_{\edits{\rm o}}$ are both Class I; then for any choice of model parameters, the indecision-breaking bifurcation in \eqref{EQ:value_dynamics} is a social dissonance-driven pitchfork bifurcation.

     iii) Suppose $\mathcal{G}_{\edits{\rm a}}$ is either Class I or Class II, $b_{ij} = 0$ for all $i\in \mathcal{V}_{\edits{\rm a}}$, $j \in \mathcal{V}_{\edits{\rm o}}$, and the indecision-breaking bifurcation in \eqref{EQ:value_dynamics} is a social dissonance-driven pitchfork bifurcation. Let $\Zz^* = (z_{11}^*,  \dots, z_{N_{\edits{\rm a}} N_{\edits{\rm o}}}^*)$ be an equilibrium on one of the new solution branches with sufficiently small $\lvert u-u^* \rvert$. The switching function $\theta_{\edits{\rm a}}$ induces a bipartition  $\mathcal{V}_{\edits{\rm a 1}},\mathcal{V}_{\edits{\rm a 2}}$, $\mathcal{V}_{\edits{\rm a 1}} \cup \mathcal{V}_{\edits{\rm a 2}}=\mathcal{V}_{\edits{\rm a}}$, $\mathcal{V}_{\edits{\rm a 1}} \cap \mathcal{V}_{\edits{\rm a 2}} = \emptyset$, of the agent set $\mathcal{V}_{\edits{\rm a}}$ such that $i\in \mathcal{V}_{\edits{\rm a 1}}$ whenever $\theta_{\edits{\rm a}}(i) = 1$ and $i\in \mathcal{V}_{\edits{\rm a 2}}$ whenever $\theta_{\edits{\rm a}}(i) = -1$. Whenever $i,k\in \mathcal{V}_p$, $p \in \{ \edits{\rm a}1,a2\}$, and $j \in \mathcal{V}_{\edits{\rm o}} \setminus \mathcal{V}_{\edits{\rm o}}^0$, $\operatorname{sign}(z_{ij}^*) = \operatorname{sign}(z_{kj}^*)$ and $\operatorname{sign}(z_{ij}^*) \neq \operatorname{sign}(z_{kj}^*)$ otherwise. If $|(\mathbf{v}_{\edits{\rm a}})_i| > |(\mathbf{v}_{\edits{\rm a}})_k|$, then $|z_{ij}^*| > |z_{kj}^*|$ for all $j \in \mathcal{V}_{\edits{\rm o}}$.

    iv) Suppose $\mathcal{G}_{\edits{\rm o}}$ is either Class I or Class II, $b_{ij} = 0$ for all $i\in \mathcal{V}_{\edits{\rm a}}$, $j \in \mathcal{V}_{\edits{\rm o}}$, and the indecision-breaking bifurcation in \eqref{EQ:value_dynamics} is a social dissonance-driven pitchfork bifurcation. Let $\Zz^* = (z_{11}^*,  \dots, z_{N_{\edits{\rm a}} N_{\edits{\rm o}}}^*)$ be an equilibrium on one of the new solution branches with sufficiently small $\lvert u-u^* \rvert$. The switching function $\theta_{\edits{\rm o}}$ induces a bipartition $\mathcal{V}_{\edits{\rm o}1},\mathcal{V}_{\edits{\rm o}2}$, $\mathcal{V}_{\edits{\rm o}1} \cup \mathcal{V}_{\edits{\rm o}2}=\mathcal{V}_{\edits{\rm o}}$, $\mathcal{V}_{\edits{\rm o}1} \cap \mathcal{V}_{\edits{\rm o}2} = \emptyset$, of options set  $\mathcal V_{\edits{\rm o}}$ such that $j\in \mathcal{V}_{\edits{\rm o}1}$ whenever $\theta_{\edits{\rm o}}(j) = 1$ and $j\in \mathcal{V}_{\edits{\rm o}2}$ whenever $\theta_{\edits{\rm o}}(j) = -1$. Whenever $j,l\in \mathcal{V}_p$, $p \in \{\edits{\rm o}1,o2\}$, and $i \in \mathcal{V}_{\edits{\rm a}} \setminus \mathcal{V}_{\edits{\rm a}}^0$, $\operatorname{sign}(z_{ij}^*) = \operatorname{sign}(z_{il}^*)$ and $\operatorname{sign}(z_{ij}^*) \neq \operatorname{sign}(z_{il}^*)$ otherwise. If $|(\mathbf{v}_{\edits{\rm o}})_j| > |(\mathbf{v}_{\edits{\rm o}})_l|$, then $|z_{ij}^*| > |z_{il}^*|$ for all $i \in \mathcal{V}_{\edits{\rm a}}$.
    
\end{corollary}

 For a social dissonance-driven bifurcation with 
 $\mathcal{G}_{\edits{\rm a}}$ and $\mathcal{G}_{\edits{\rm o}}$ in Class I or II, Corollary~\ref{cor:pitchfork_conds} 
 predicts relative signs and ordering of strength of equilibrium opinions, the latter according to the ordering of components in the graph eigenvector $\mathbf{v}_{\edits{\rm a}}$ or $\mathbf{v}_{\edits{\rm o}}$.  According to the Corollary, each graph belonging to Class I and/or Class II of Definition \ref{def:graph-classes} is a sufficient condition for a pitchfork bifurcation. Importantly, these are not \textit{necessary} conditions as we already observed with Example \ref{ex:pitchfork} (Figs.~\ref{fig:pitchfork}-\ref{fig:EQ_options}). The belief system graph $\mathcal{G}_{\edits{\rm o}}$ of Fig.~\ref{fig:graphs}(b) used in Example~\ref{ex:pitchfork} is not digon symmetric. 
$\mathcal{G}_{\edits{\rm o}}$ is also not eventually structurally balanced: $A_{\edits{\rm o}}^2$ and $A_{\edits{\rm o}}^3 = I$ have zero entries, and $A_{\edits{\rm o}}^4=A_{\edits{\rm o}}$. Yet the bifurcation observed in Example \ref{ex:pitchfork} is a social dissonance-driven pitchfork bifurcation. In general, when the structure of social network graph $\mathcal{G}_{\edits{\rm a}}$ and belief system graph $\mathcal{G}_{\edits{\rm o}}$ are known, the conditions of Theorem \ref{thm:pitchfork-No} are easily verified by computing eigenvalues of adjacency matrices $A_{\edits{\rm a}},A_{\edits{\rm o}}$. 

Nonetheless, the sufficient conditions of Corollary \ref{cor:pitchfork_conds} are particularly meaningful in the context of belief formation. Structural balance in a graph corresponds to a bipartition of the nodes, with strictly positive edges within each subgroup and strictly negative edges between subgroups. For a communication graph this bipartition corresponds to the split of a group into two mutually antagonistic subgroups, as may occur in a polarized social network. For a belief system graph, structural balance may correspond, for example, to a partition of political issues into two subgroups based on their alignment with left-leaning or right-leaning ideological views. Corollary \ref{cor:pitchfork_conds}.iv is additionally related to the notion of \textit{coherence} of beliefs \cite{converse2006nature,introne2023measuring}. An individual holds a maximally coherent set of beliefs when their beliefs \edits{do not} violate the logical constraints imposed by the belief system $\mathcal{G}_{\edits{\rm o}}$. According to Corollary \ref{cor:pitchfork_conds}.iv, when the graph $\mathcal{G}_{\edits{\rm o}}$ is structurally balanced, all individuals on the network form maximally coherent beliefs, as their belief vector at steady-state directly reflects the bipartition of beliefs encoded in the structure of $\mathcal{G}_{\edits{\rm o}}$.

\begin{figure}
    \centering
    \includegraphics[width = 0.9\columnwidth]{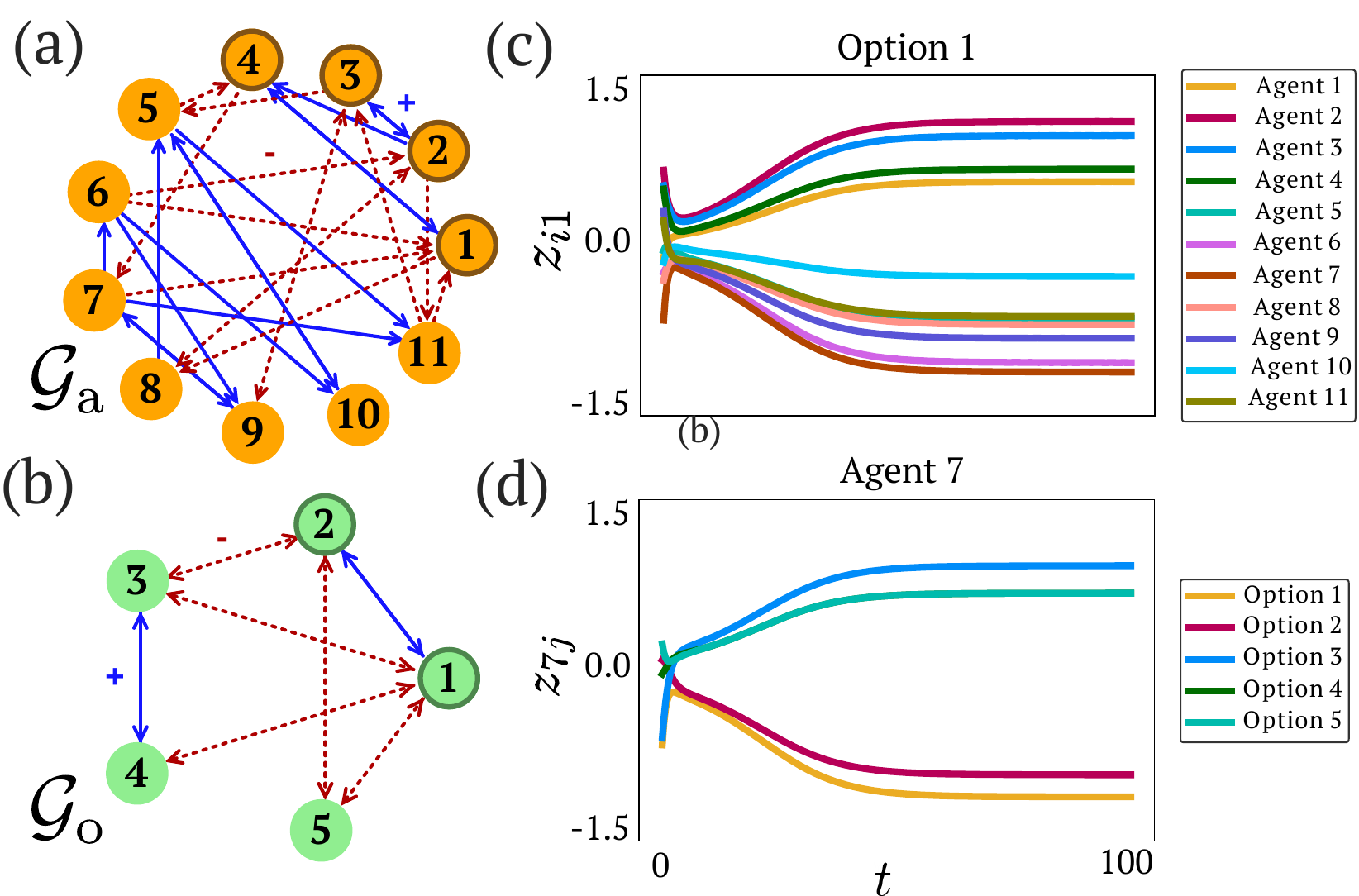}
    \caption{(a) Strongly connected, structurally balanced communication graph with node partitions $\mathcal{V}_{\edits{\rm a 1}} = \{ 1, 2, 3, 4\}$, $\mathcal{V}_{\edits{\rm a 2}} = \{ 5,6,7,8,9,10,11\}$;   (b) Strongly connected, structurally balanced belief system graph with node partitions $\mathcal{V}_{\edits{\rm o}1} = \{1,2\}$, $\mathcal{V}_{\edits{\rm o}2} = \{3,4,5\}$; (c) belief trajectories of all 11 agents on option 1; (d) belief trajectories of agent 7 on all 5 options; belief trajectories on options 4 and 5 overlap for most of the simulation. Parameters: $u = u^* + 0.05 \approx 0.7627$, $\alpha = \gamma = \beta = \delta = 0.1$, $d = 1$, $b_{ij} = 0$ for all $i \in \mathcal{V}_{\edits{\rm a}}$, $j \in \mathcal{V}_{\edits{\rm o}}$, $S_1(\cdot) = \tanh(\cdot)$, $S_2(\cdot) = \frac{1}{2} \tanh(2 \cdot)$.  }
    \label{fig:balanced-pitchfork}
\end{figure}

\begin{example} \label{ex:balanced-pitchfork}
Consider the belief dynamics model \eqref{EQ:value_dynamics} for 11 agents forming beliefs on 5 options with structurally balanced $\mathcal{G}_{\edits{\rm a}}$, $\mathcal{G}_{\edits{\rm o}}$ pictured in Fig. \ref{fig:balanced-pitchfork}(a),(b). With the indicated choice of parameters, the system satisfies the conditions of Theorem \ref{thm:pitchfork-No} and exhibits a social dissonance-driven pitchfork bifurcation. According to Corollary \ref{cor:pitchfork_conds}.iii, due to the structural balance of the communication graph we expect the beliefs of agents  1-4  on each option to share sign at steady-state, while agents 5-11 will form same-sign beliefs that are opposite-sign from agents 1-4 following the bipartition in $\mathcal{G}_{\edits{\rm a}}$. Belief trajectories simulated from a random initial condition in Fig. \ref{fig:balanced-pitchfork}(c) confirm this prediction. Analogously, according to Corollary \ref{cor:pitchfork_conds}.iv we expect each agent at steady-state to simultaneously favor options 1 and 2 while rejecting the rest, or simultaneously favor options 3-5 while rejecting 1 and 2 following the bipartition in $\mathcal{G}_{\edits{\rm o}}$. This prediction is confirmed by the simulated belief trajectories in Fig. \ref{fig:balanced-pitchfork}(d). 
\end{example}

\section{Oscillatory belief-formation: Hopf bifurcation \label{sec:Hopf}}

So far we have explored the emergence of multistable belief equilibria through pitchfork bifurcations, a generic outcome of indecision-breaking in the belief formation model \eqref{EQ:value_dynamics}. We now study under which conditions a belief-forming bifurcation can lead to an oscillatory belief state. Oscillations in belief systems have received scant attention to date, both in formal models and in empirical studies of social systems. Within an individual, belief oscillations are thought by some to be an important dynamic feature of cognition that plays a role in decision-making and attitude change, a view that is supported by indirect experimental evidence \cite{fink1993oscillation,fink2002oscillation,chung2012sequential}. At the group level, oscillations in beliefs and preferences are commonly observed in human societies, e.g., in recurring fashion trends \cite{belleau1987cyclical} and in recurring swings between conservative and liberal political attitudes \cite{stimson2018public}. We derive sufficient and necessary conditions for belief oscillations that arise from a Hopf bifurcation in the belief formation model \eqref{EQ:value_dynamics}. \edits{Oscillations are also interesting from the point of view of design of social behavior, as they allow agents to explore different options dynamically. The model \eqref{EQ:value_dynamics} supports both oscillations and equilibria, thus allowing for systematic design of transitions between oscillitary and stationary-belief operating regimes in networks of artificial agents. }The results in this section are partially adapted from the recent conference paper \cite{bizyaeva2022sustained}.

A Hopf bifurcation is a local bifurcation in which a limit cycle appears near a singular equilibrium of a nonlinear system. Consider the belief formation model \eqref{EQ:value_dynamics}. A necessary condition for the indecision-breaking bifurcation point $(\Zz,u) = (\mathbf{0},u^*)$ of Theorem \ref{thm:bif} to be a \textit{Hopf bifurcation point} is for the Jacobian  $J(\mathbf{0},u)$ \eqref{eq:jac} to have a complex conjugate pair of leading eigenvalues $\eta_{+}(u) = \eta_r(u) + i \eta_i(u)$, $\eta_{-}(u) = \overline{\eta}_{+}(u) = \eta_r(u) - i \eta_i(u)$, with $\eta_+(u^*) = i \eta_i(u^*)$ and $\eta_-(u^*) = - i \eta_i(u^*)$ being the only eigenvalues of $J(\mathbf{0},u^*)$ on the imaginary axis \cite[Theorem 3.4.2]{guckenheimer2013nonlinear}. To enforce this necessary condition we introduce the following assumption.

\begin{assumption}[Complex-conjugate leading eigenvalues]\label{ass1}
The set $\Lambda$ of leading eigenvalues of $J(\mathbf{0},u)$, defined in \eqref{eq:Lambda}, satisfies one of the following three conditions:

1) \textit{Complex eigenvalue in belief system graph:} $\Lambda = \{(\lambda,\mu),(\lambda,\overline{\mu})\}$ where $\lambda \in \R{}$, $\mu \in \mathbb{C}$, $\lambda, \operatorname{Re}(\mu),\operatorname{Im}(\mu) \neq 0$;

2) \textit{Complex eigenvalue in social network graph:} $\Lambda = \{(\lambda,\mu),(\overline{\lambda},\mu)\}$ where $\lambda \in \mathbb{C}$, $\mu \in \R{}$, $\mu, \operatorname{Re}(\lambda),\operatorname{Im}(\lambda) \neq 0$;

3) \textit{Interaction of complex eigenvalues:} $\Lambda = \{ (\lambda,\mu),(\overline{\lambda},\overline{\mu})\}$ where $\lambda,\mu \in \mathbb{C}$, $\operatorname{Re}(\lambda)$, $\operatorname{Im}(\lambda)$, $\operatorname{Re}(\mu)$, $\operatorname{Im}(\mu) \neq 0$.
\end{assumption}

\begin{theorem}[Hopf bifurcation] \label{thm:hopf}
\cite[Theorem IV.3]{bizyaeva2022sustained}
Consider \eqref{EQ:value_dynamics} with communication graph $\mathcal{G}_{\edits{\rm a}}$ and belief system graph $\mathcal{G}_{\edits{\rm o}}$. Let Assumption \ref{ass1} hold. Suppose $\lambda^{\dagger}\in\sigma(A_{\edits{\rm a}}), \mu^{\dagger} \in \sigma(A_{\edits{\rm o}})$ generate $\eta_{+}(u)$, i.e. $\eta_+(u) = \eta(u,\lambda^\dagger,\mu^\dagger)$, and assume $\alpha + \gamma \operatorname{Re}(\lambda^\dagger) + \beta \operatorname{Re}(\mu^\dagger) + \delta \operatorname{Re}(\lambda^\dagger \mu^\dagger) > 0$. Let $\mathbf{w}_{\edits{\rm a}},\mathbf{v}_{\edits{\rm a}} \in \mathds{C}^{N_{\edits{\rm a}}}$ be the left and right eigenvectors of $A_{\edits{\rm a}}$ corresponding to $\lambda^{\dagger}$ and $\overline{\lambda^{\dagger}}$, respectively; let $\mathbf{w}_{\edits{\rm o}},\mathbf{v}_{\edits{\rm o}}\in \mathbb{C}^{N_{\edits{\rm o}}}$ be the left and right eigenvectors of $A_{\edits{\rm o}}$ corresponding to $\mu^{\dagger}$ and $\overline{\mu^{\dagger}}$, respectively, satisfying the biorthogonal normalization condition 
\begin{equation*}
    \langle \mathbf{w}_{\edits{\rm a}} \otimes \mathbf{w}_{\edits{\rm o}}, \mathbf{v}_{\edits{\rm a}} \otimes \mathbf{v}_{\edits{\rm o}}, \rangle = 2, \ \ \langle \overline{\mathbf{w}_{\edits{\rm a}} \otimes \mathbf{w}_{\edits{\rm o}}}, \mathbf{v}_{\edits{\rm a}} \otimes \mathbf{v}_{\edits{\rm o}}, \rangle = 0. \label{eq:eigenvector-normalization-hopf}
\end{equation*}
For sufficiently small  $|g_1''(0)|$, $|g_2''(0)|$, the following hold:

1) There is a unique 3-dimensional center manifold $W^c \subset \mathbb{R}^{N_{\edits{\rm a}} N_{\edits{\rm o}}} \times \R{}$ passing through $(\Zz,u) = (\mathbf{0},u^*)$, tangent to $\operatorname{span}\{ \operatorname{Re}(\mathbf{v}_{\edits{\rm a}} \otimes \mathbf{v}_{\edits{\rm o}}), \operatorname{Im}(\mathbf{v}_{\edits{\rm a}} \otimes \mathbf{v}_{\edits{\rm o}}) \} $ at $u = u^* = d/(\alpha + \gamma \operatorname{Re}(\lambda^\dagger) + \beta \operatorname{Re}(\mu^\dagger) + \delta \operatorname{Re}(\lambda^\dagger \mu^\dagger))$. There is a family of periodic orbits of \eqref{EQ:value_dynamics} that bifurcates from the neutral equilibrium $\Zz = \mathbf{0} $ along $W_c$ at $u = u^*$;

2) Let $K = \operatorname{Re}\big( \big( S_1'''(0) \left(\alpha + \gamma \lambda^{\dagger}\right) \left| \alpha + \gamma \lambda^{\dagger}\right|^2 +
+ S_2'''(0) \left(\beta + \delta  \lambda^{\dagger}\right)\mu^{\dagger}\left| \beta +  \delta \lambda^{\dagger}\right|^2 \left| \mu^\dagger \right|^2 \big)  
 \langle \mathbf{w}_{\edits{\rm a}}\otimes \mathbf{w}_{\edits{\rm o}} , \lvert\mathbf{v}_{\edits{\rm a}}\otimes \mathbf{v}_{\edits{\rm o}}\rvert^2 \odot (\mathbf{v}_{\edits{\rm a}}\otimes \mathbf{v}_{\edits{\rm o}}) \rangle \big)$
where $\lvert\mathbf{x}\rvert^2 = \overline{\mathbf{x}} \odot \mathbf{x}$. Whenever $K < 0$ the periodic solutions appear supercritically (for $u > u^*$) and are locally asymptotically stable; whenever $K > 0$, the solutions appear subcritically (for $u < u^*$) and are unstable;

3) When $\lvert u - u^* \rvert$ is small, the period of the solutions $\Zz^*(t)$ is near $T = 2 \pi /(u^* | \gamma \operatorname{Im}(\lambda^\dagger) + \beta \operatorname{Im}(\mu^{\dagger}) + \delta \operatorname{Im}(\lambda^{\dagger}\mu^{\dagger})|)$, the difference in phase between $z_{ij}^*(t)$ and $z_{kl}^*(t)$ is near $ \varphi_{ik}^{jl} = \operatorname{arg}((\mathbf{v}_{\edits{\rm a}})_i (\mathbf{v}_{\edits{\rm o}})_j) - \operatorname{arg}((\mathbf{v}_{\edits{\rm a}})_k (\mathbf{v}_{\edits{\rm o}})_l)$, and the amplitude of $z_{ij}^*(t)$ is greater than the amplitude of $z_{kl}^*(t)$ if and only if $\lvert (\mathbf{v}_{\edits{\rm a}})_i \rvert \lvert (\mathbf{v}_{\edits{\rm o}})_j\rvert > \lvert (\mathbf{v}_{\edits{\rm a}})_k \rvert \lvert (\mathbf{v}_{\edits{\rm o}})_l\rvert $.
\end{theorem}
\edits{Assumption \ref{ass1} implies that $\delta>0$ and/or $\beta,\gamma>0$.  Theorem \ref{thm:hopf} holds for general signed graphs satisfying the stated assumptions, which includes graphs that are neither Class I or II, and ones that are not strongly connected.}

\begin{figure}
    \centering
    \includegraphics[width=\columnwidth]{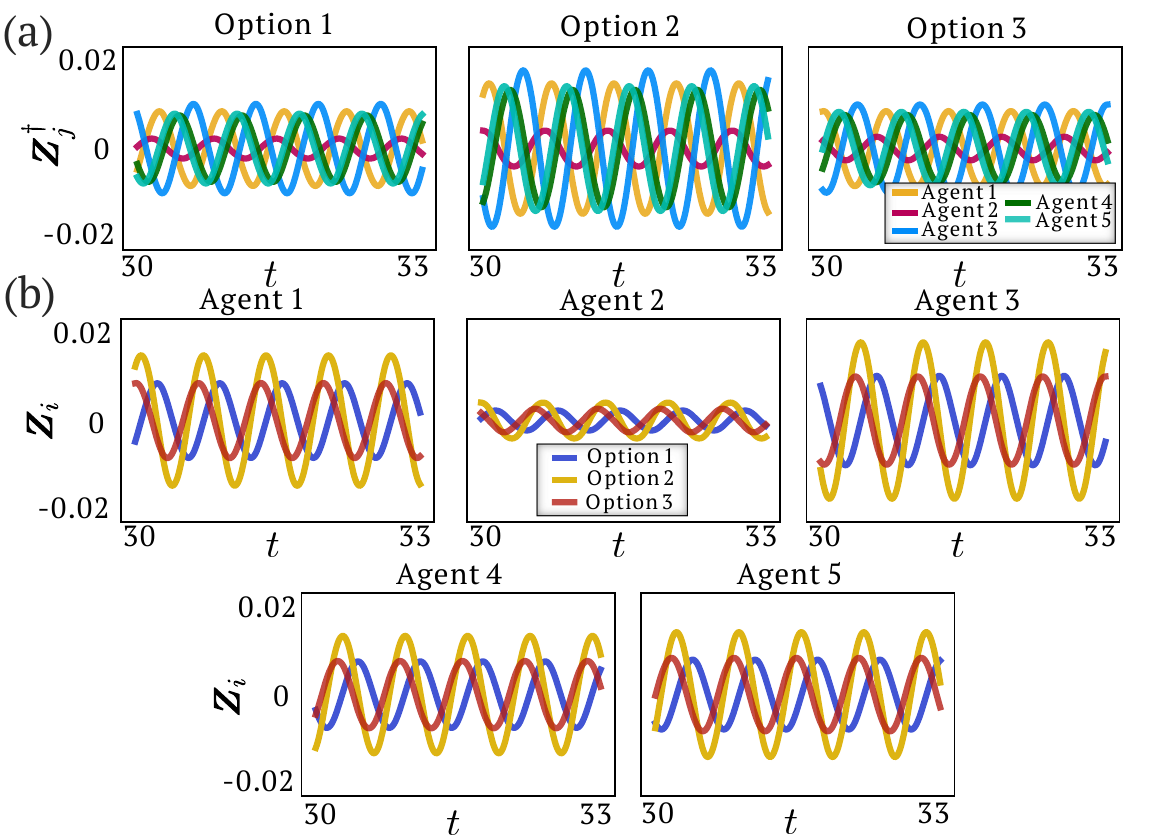}
    \caption{Trajectories $z_{ij}(t)$ of stable oscillation described in Example \ref{ex:hopf} for \eqref{EQ:value_dynamics} on communication and belief system graphs of Fig.~\ref{fig:graphs}; (a) grouped by option; (b) grouped by agent. Parameters: $d = 10$, $u = 1$, $\alpha = \gamma = \beta = 0.1$, $\delta = 12$, $b_{ij} = 0$ for all $i \in \mathcal{V}_{\edits{\rm a}}$, $j \in \mathcal{V}_{\edits{\rm o}}$, $S_1(\cdot) = \tanh(\cdot)$, $S_2(\cdot) = \frac{1}{2}\tanh(2 \ \cdot)$.}
    \label{fig:Hopf}
\end{figure}

\begin{example} \label{ex:hopf} To illustrate the predictions of Theorem \ref{thm:hopf} we use again the communication and belief system graphs $\mathcal{G}_{\edits{\rm a}}$,$\mathcal{G}_{\edits{\rm o}}$ of Fig. \ref{fig:graphs}. Recall from Example \ref{ex:GaGo} that for these graphs the set of graph eigenvalue pairs $(\lambda,\mu) \in \sigma(A_{\edits{\rm a}}) \times \sigma(A_{\edits{\rm o}})$ that maximize $\operatorname{Re}(\lambda \mu)$ has the structure $\Lambda_2 = \{(\lambda^\dagger, \mu^\dagger), (\overline{\lambda}^\dagger,\overline{\mu}^\dagger) \}$ and satisfies the third case in Assumption \ref{ass1}. It follows from Theorem \ref{thm:hopf} and Proposition \ref{prop:param-effects-2} that any external dissonance-dominant bifurcation on these graphs will be a Hopf bifurcation, and this bifurcation is realized for sufficiently large values of the external dissonance parameter $\delta$. Consider \eqref{EQ:value_dynamics} on these graphs with the parameters as in Fig.~\ref{fig:Hopf}. The periodic orbits bifurcate supercritically and are stable. The ordering of phases and amplitudes of these oscillations along each option, as in Fig.~\ref{fig:Hopf}(a), and internally, as in Fig.~\ref{fig:Hopf}(b) are predicted by the complex eigenvectors $\mathbf{v}_{\edits{\rm a}}$ and $\mathbf{v}_{\edits{\rm o}}$.
\end{example}

The three types of eigenvalue combinations in Assumption \ref{ass1} lead to qualitatively different types of oscillation. Case 1 (Case 2) lead to phase-synchronized intra-agent (inter-agent) beliefs, i.e., the components of $\Zz_i(t)$ ($\Zz_j^\dagger(t)$) oscillate with the same phase for each $i$ (each $j$). Case 3 generically leads to more complex oscillatory phase patterns, unless some kind of symmetry is present in one of the graphs.
The following corollary formalizes these observations.  Let Class I and Class II be defined as in Definition \ref{def:graph-classes}. Given graph $\mathcal{G}$ with adjacency matrix $A$ that is either Class I or II, let $\theta: \mathcal{V} \to \{1,-1\}$ be the switching function with switching matrix $\Theta$, and let $P_n$ be a permutation matrix for which $\Theta P_n A P_n^T\Theta$ is eventually positive if $\mathcal{G}$ is Class I or nonnegative if $\mathcal{G}$ is Class II.  \edits{Recall that \textit{social dissonance-driven} and \textit{external dissonance-driven} bifurcations happen when eigenvalue sets $\Lambda_1$ and $\Lambda_2$, respectively, are dominant (see Sections \ref{sec:params},\ref{sec:NB_interp}).} 

\begin{corollary} \label{cor:balanced-oscillations}
    i.1) Suppose that $\mathcal{G}_{\edits{\rm a}}$ is either Class I or Class II with associated dominant eigenvalue $\lambda_{max}$ and corresponding eigenvector $\mathbf{v}_{\edits{\rm a}}$,  and that $A_o$ has a complex conjugate pair of leading eigenvalues $\mu, \overline{\mu}$. Then, a social dissonance-driven indecision-breaking bifurcation in \eqref{EQ:value_dynamics} is a Hopf bifurcation. For sufficiently small $\lvert u-u^* \rvert$ and $\mathbf{b} = \mathbf{0}$, let $\Zz^*(t) = (z_{11}^*(t),  \dots, z_{N_{\edits{\rm a}} N_{\edits{\rm o}}}^*(t))$ be the periodic solution of \eqref{EQ:value_dynamics} emerging at the Hopf bifurcation.
    
    i.2)  Let $\mathcal{V}_{\edits{\rm a}}^0 = \{i \in \mathcal{V}_{\edits{\rm a}} \ \ s.t. \ \ (\mathbf{v}_{\edits{\rm a}})_i = 0 \} \subset \mathcal{V}_{\edits{\rm a}}$. Also let $\mathcal{V}_{\edits{\rm a 1}},\mathcal{V}_{\edits{\rm a 2}}\subset\mathcal{V}_{\edits{\rm a}}$, with $\mathcal{V}_{\edits{\rm a 1}} \cup \mathcal{V}_{\edits{\rm a 2}}=\mathcal{V}_{\edits{\rm a}}$, $\mathcal{V}_{\edits{\rm a 1}} \cap \mathcal{V}_{\edits{\rm a 2}} = \emptyset$ be the bipartition of $\mathcal{V}_{\edits{\rm a}}$ induced by the switching function $\theta_{\edits{\rm a}}$ such that $i\in \mathcal{V}_{\edits{\rm a 1}}$ whenever $\theta_{\edits{\rm a}}(i) = 1$ and $i\in \mathcal{V}_{\edits{\rm a 2}}$ whenever $\theta_{\edits{\rm a}}(i) = -1$.
    Suppose $\|(\Zz^\dagger_j)^*(t) \| > 0$ for some $t\in\R{}$ and $j \in \mathcal{V}_{\edits{\rm o}} \setminus \mathcal{V}_{\edits{\rm o}}^0$. If $i,k\in \mathcal{V}_p$, $p \in \{a1,a2\}$, then $\operatorname{sign}(z_{ij}^*(t)) = \operatorname{sign}(z_{kj}^*(t))$. Conversely, if $i\in \mathcal{V}_{p}$ and $k\in \mathcal{V}_{s}$, $p,s \in \{\edits{\rm a}_1,\edits{\rm a}_1\}$, $p \neq s$, then $\operatorname{sign}(z_{ij}^*(t)) \neq \operatorname{sign}(z_{kj}^*(t))$ .

    ii.1)  Suppose that $\mathcal{G}_{\edits{\rm o}}$ is either Class I or Class II with associated dominant eigenvalue $\mu_{max}$ and corresponding eigenvector $\mathbf{v}_{\edits{\rm o}}$, and that $\mathcal{G}_{\edits{\rm a}}$ has a complex conjugate pair of leading adjacency eigenvalues $\lambda,\overline{\lambda}$. Then a social dissonance-driven indecision-breaking bifurcation in \eqref{EQ:value_dynamics} is a Hopf bifurcation. For sufficiently small $\lvert u-u^* \rvert$ and $\mathbf{b} = \mathbf{0}$, let $\Zz^*(t) = (z_{11}^*(t),  \dots, z_{N_{\edits{\rm a}} N_{\edits{\rm o}}}^*(t))$ be a periodic solution of \eqref{EQ:value_dynamics} emerging at the Hopf bifurcation. 
    
    ii.2) Let $\mathcal{V}_{\edits{\rm o}}^0 = \{j \in \mathcal{V}_{\edits{\rm o}} \ \ s.t. \ \ (\mathbf{v}_{\edits{\rm o}})_j = 0 \} \subset \mathcal{V}_{\edits{\rm a}}$. Also let $\mathcal{V}_{\edits{\rm o}1},\mathcal{V}_{\edits{\rm o}2}$, $\mathcal{V}_{\edits{\rm o}1} \cup \mathcal{V}_{\edits{\rm o}2}=\mathcal{V}_{\edits{\rm o}}$, $\mathcal{V}_{\edits{\rm o}1} \cap \mathcal{V}_{\edits{\rm o}2} = \emptyset$ be the bipartition of options induced by the switching function $\theta_{\edits{\rm o}}$ such that $j\in \mathcal{V}_{\edits{\rm o}1}$ whenever $\theta_{\edits{\rm o}}(j) = 1$ and $j\in \mathcal{V}_{\edits{\rm o}2}$ whenever $\theta_{\edits{\rm o}}(j) = -1$. Suppose $\|\Zz_i^*(t) \| > \varepsilon$ for some $t \in \mathbb{R}$, $i \in \mathcal{V}_{\edits{\rm a}} \setminus \mathcal{V}_{\edits{\rm a}}^0$, and some small $\varepsilon$. If $j,l\in \mathcal{V}_p$, $p \in \{o1,o2\}$, then $\operatorname{sign}(z_{ij}^*(t)) = \operatorname{sign}(z_{il}^*(t))$. Conversely, if $j \in \mathcal{V}_p$, $l \in \mathcal{V}_s$, $p,s \in \{\edits{\rm a}_1,\edits{\rm a}_2\}$, $p \neq s$, then $\operatorname{sign}(z_{ij}^*(t)) \neq \operatorname{sign}(z_{il}^*(t))$. 
\end{corollary}

As a consequence of Corollary \ref{cor:balanced-oscillations}, social dissonance-driven oscillations on structurally balanced communication graphs reflect the partition of a graph into two polarized subgroups because the beliefs of agents in the two subgroups oscillate exactly out-of-phase with one another, i.e., the agents in the two subgroups will always disagree with each other despite the oscillatory trajectories of their beliefs. Analogously, social dissonance-driven oscillations on structurally balanced belief system graphs will reflect the induced bipartition of options in the internal ordering of each agent's beliefs. In other words, despite their oscillatory trajectories, each agent's beliefs remain coherent at all times and do not violate the logical relationships encoded in the structure of belief system $\mathcal{G}_{\edits{\rm o}}$.

\begin{figure}
    \centering
    \includegraphics[width=\columnwidth]{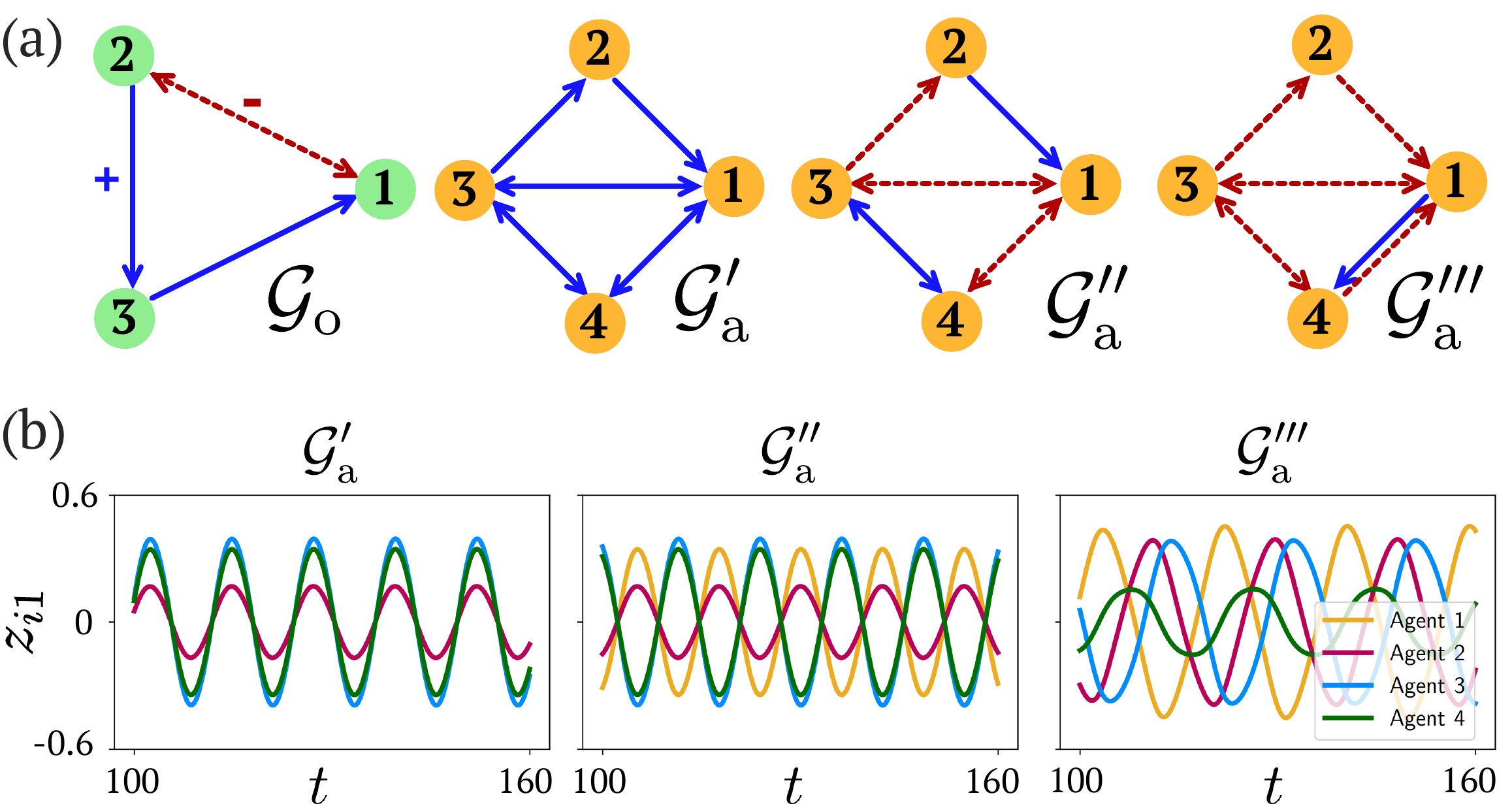}
    \caption{Figure adapted from \cite{bizyaeva2022sustained}. (a) Belief system graph $\mathcal{G}_{\edits{\rm o}}$ and three communication graphs $\mathcal{G}_{\edits{\rm a}}'$,$\mathcal{G}_{\edits{\rm a}}''$,$\mathcal{G}_{\edits{\rm a}}'''$. (b) Trajectories of belief dynamics for all four agents on option 1 with belief system graph $\mathcal{G}_{\edits{\rm o}}$ and each of the three communication graphs shown in (a). Parameters:  $d = 1$, $\alpha = \gamma = 0.1$, $\beta = \delta = 0.25$, $b_{ij} = 0$ for all $i \in \mathcal{V}_{\edits{\rm a}}$, $j \in \mathcal{V}_{\edits{\rm o}}$; $u = 1.25$ for $\mathcal{G}_{\edits{\rm a}}'$ and $\mathcal{G}_{\edits{\rm a}}''$; $u = 1.7$ for $\mathcal{G}_{\edits{\rm a}}'''$; $S_1(\cdot) = \tanh(\cdot)$, $S_2(\cdot) = \frac{1}{2}\tanh(2 \ \cdot)$. 
    }
    \label{fig:3oscillations}
\end{figure}
\begin{example} \label{ex:3oscillations}
    Consider the belief dynamics \eqref{EQ:value_dynamics} with belief system graph $\mathcal{G}_{\edits{\rm o}}$ and three different communication graphs, pictured in Fig. \ref{fig:3oscillations}(a) with parameter values indicated in the figure caption. The leading eigenvalues of $A_{\edits{\rm o}}$ are a complex-conjugate pair. The graphs $\mathcal{G}_{\edits{\rm a}}'$ and $\mathcal{G}_{\edits{\rm a}}''$ are both Class II and switching equivalent, with $A_{\edits{\rm a}}'$ having the Perron-Frobenius property while $A_{\edits{\rm a}}''$ is structurally balanced with a bipartition of nodes $\mathcal{V}_{\edits{\rm a 1}} = \{1,2\}$, $\mathcal{V}_{\edits{\rm a 2}} = \{3,4\}$. Following Corollary \ref{cor:balanced-oscillations}.i, in a social dissonance-driven Hopf bifurcation, beliefs of agents on $\mathcal{G}_{\edits{\rm a}}'$ synchronize in phase along every option while beliefs of agents 1 and 2 on $\mathcal{G}_{\edits{\rm a}}''$ oscillate out of phase from agents 3 and 4, reflecting the bipartition. Sample simulated trajectories are shown in Fig.~\ref{fig:3oscillations}(b) for option 1. Meanwhile, $\mathcal{G}_{\edits{\rm a}}''$ is neither class I or II, and has an adjacency matrix with a complex conjugate pair of leading eigenvalues. A social dissonance-driven bifurcation on this graph is still a Hopf bifurcation; however, due to the fact that both $\mathbf{v}_{\edits{\rm a}}$ and $\mathbf{v}_{\edits{\rm o}}$ of Theorem \ref{thm:hopf} are complex, beliefs do not synchronize in phase internally or between agents along any option. A sample trajectory for option 1 is shown in Fig.~\ref{fig:3oscillations}(b).
\end{example}

We next derive sufficient conditions for  graphs $\mathcal{G}_{\edits{\rm a}}$, $\mathcal{G}_{\edits{\rm o}}$ that either rule out oscillations altogether or guarantee oscillations when social dissonance is dominant in the belief dynamics \eqref{EQ:value_dynamics}.

\begin{corollary} \label{cor:graph-cond-Hopf}
1) \cite[Proposition V.1]{bizyaeva2022sustained} Suppose $\mathcal{G}_{\edits{\rm a}}$ and $\mathcal{G}_{\edits{\rm o}}$ are both undirected. Then the indecision-breaking bifurcation of Theorem \ref{thm:bif} \textit{cannot} be a Hopf bifurcation. 

2) Suppose $\mathcal{G}_{\edits{\rm a}}$ ($\mathcal{G}_{\edits{\rm o}}$) belongs to Class I and/or Class II of Definition \ref{def:graph-classes}. Suppose $\mathcal{G}_{\edits{\rm o}}$ ($\mathcal{G}_{\edits{\rm a}}$) has an odd number of nodes and is switching isomorphic to a cycle graph $\mathcal{G}$ with all-negative edge signatures, i.e.,  $\mathcal{G}$ has circulant adjacency matrix $A$ with the first row $(0, 0,\dots, 0, -1)$, and every subsequent row generated by shifting the preceding row to the right by one element. Any social dissonance-driven indecision-breaking bifurcation in \eqref{EQ:value_dynamics} on these graphs is a Hopf bifurcation.
\end{corollary}



\section{Discussion and final remarks \label{sec:final}}

We have proposed a nonlinear, multi-dimensional belief formation model and shown how agents \edits{become opinionated} as a consequence of an indecision-breaking bifurcation determined by the communication network and belief system network. We have described the multi-stable belief equilibria and belief oscillations that form through a pitchfork bifurcation and a Hopf bifurcation, respectively. 

While these two outcomes of belief formation are most common across all graphs, other more complex bifurcations can occur in the model. 
For example, the set $\Lambda_2$ described in parts iii and iv of Lemma \ref{lem:graph-cond} corresponds to two equal real-valued eigenvalues in the Jacobian $J(0,\mathbf{u})$. External dissonance-driven bifurcations on these graphs will lead to a different bifurcation than those described in Sections \ref{sec:pitchfork} and \ref{sec:Hopf}, and belief trajectories will settle on a 2-dimensional attractor. The  dissertation \cite{AnastasiaThesis} contains an example of graphs $\mathcal{G}_{\edits{\rm a}}$,$\mathcal{G}_{\edits{\rm o}}$ on which the social dissonance-driven bifurcation results in belief trajectories settling on an ergodic torus in phase space. The analysis in \cite{franci2022breaking} considers bifurcations with a belief system graph $\mathcal{G}_{\edits{\rm o}}$ that is all-to-all connected and same-sign. In this case, $\Lambda$ can consist of $N_{\edits{\rm o}}$ eigenvalue pairs, and many new symmetric equilibria can emerge at the corresponding equivariant bifurcation. When symmetry is present in one of the graphs $\mathcal{G}_{\edits{\rm a}}$,$\mathcal{G}_{\edits{\rm o}}$, tools from equivariant bifurcation theory \cite{GolubitskySymmetryPerspective} can be used to classify the indecision-breaking bifurcation.

The proposed model 
can thus be understood to parametrize a set of belief-forming behaviors that are extremely rich but nevertheless interpretable with a small number of 
parameters that capture realistic effects. Our analyses provide novel insights about the dynamics of social systems. They also serve as building blocks, as in \cite{HaiminCDC2023,bizyaeva2022switching,cathcart2022opinion,paine2023model},  for new classes of decentralized algorithms for decision-making, task allocation, and other types of control for engineered teams and networks. \edits{Future work will also aim to study the interplay between the dynamics of beliefs and time-varying networks, with particular emphasis on balanced networks, as preliminarily explored in~\cite{bizyaeva2022switching}.}

\section*{Appendix} 

\subsection*{Alternative model formulations}

The model \eqref{eq:dynamics_with_observation} can be formulated alternatively as
\begin{gather}
   \dot{z}_{ij}=-d_i \ z_{ij} + 
	u_iS\bigg( \alpha_i z_{ij} +b_{ij} + \gamma_i  \textstyle\sum_{\substack{k=1 \\ k \neq i}}^\Na(A_{\edits{\rm a}})_{ik} y_{ikj} \nonumber  \\
	 + \beta_i  \textstyle\sum_{\substack{l\neq j\\l=1}}^\No (A_{\edits{\rm o}})_{jl} z_{il} + \delta_i \textstyle\sum_{\substack{l\neq j\\l=1}}^\No \textstyle\sum_{\substack{k = 1\\k \neq i}}^\Na (A_{\edits{\rm o}})_{jl} (A_{\edits{\rm a}})_{ik} y_{ikl} \bigg)  \nonumber 
\end{gather}
where $S:\R{}\to[-k_1,k_2]$ with $k_1,k_2 > 0$ is a saturating function $S(y) = \hat{S}(y) + g(y)$ satisfying $S(0) = 0$, $\hat{S}(-y) = - \hat{S}(y)$, $\hat{S}'(0) = 1$,  $g(0) = 0$, $g'(0) = 0$, $g(-y) \neq -g_m(y)$. In the alternative model above, agents process all dissonance terms together, whereas in the model \eqref{eq:dynamics_with_observation} saturation is applied to social information along each option separately. 

Another reasonable assumption may be that belief states are saturated individually, limiting the influence of any single strong belief. This is captured in a third  model formulation:
\begin{gather}
   \dot{z}_{ij}=-d_i \ z_{ij} + b_{ij} +
	u_i \bigg( \alpha_i S(z_{ij}) + \gamma_i  \textstyle\sum_{\substack{k=1 \\ k \neq i}}^\Na(A_{\edits{\rm a}})_{ik} S(y_{ikj}) \label{eq:model-alternative2} \nonumber \\
	 + \beta_i  \textstyle\sum_{\substack{l\neq j\\l=1}}^\No (A_{\edits{\rm o}})_{jl} S(z_{il}) + \delta_i \textstyle\sum_{\substack{l\neq j\\l=1}}^\No \textstyle\sum_{\substack{k = 1\\k \neq i}}^\Na (A_{\edits{\rm o}})_{jl} (A_{\edits{\rm a}})_{ik} S(y_{ikl})\bigg). \nonumber 
\end{gather}

Due to the local and topological nature of the bifurcation analysis in this paper, all of the results proved for \eqref{eq:dynamics_with_observation} apply directly to the two alternative models above. The three models can be used interchangeably, and selected based on the interpretation of nonlinear processing that best fits an application context.  
The magnitudes of steady-state opinions, given the same set of model parameters, will differ between \eqref{eq:dynamics_with_observation}, \eqref{eq:model-alternative}, and \eqref{eq:model-alternative2} by some nonlinear scaling; however, the local topological features of the flow near the indecision-breaking bifurcation point will be shared by the three models.

\subsection*{Specialization to one topic}

All of the results in this paper apply to the special case $N_{\edits{\rm o}} = 1$, in which agents form scalar beliefs $\Zz_i = z_{i1} \in \mathds{R}$ and the belief system graph $\mathcal{G}_{\edits{\rm o}}$ is trivial. We denote  $z_{i1}$ by $z_i$,  $y_{ik1}$ by $y_{ik}$, and $b_{i1}$ by $b_i$; the dynamics \eqref{eq:dynamics_with_observation} are 
\begin{equation}
    \dot{z}_i = - d \ z_i + u \  S_{1}\left( \alpha z_{i} + \gamma \textstyle \sum_{\substack{k=1 \\ k \neq i}}^\Na(A_{\edits{\rm a}})_{ik} y_{ik}\right) + b_{i}. \label{eq:scalar-dyn}
\end{equation}
The results in Theorems \ref{thm:bif}, \ref{thm:pitchfork-No}, \ref{thm:hopf}, and their various corollaries specialize to \eqref{eq:scalar-dyn} with perfect observation $y_{ik} = z_k$ with the substitutions $A_{\edits{\rm o}} =\mathbf{v}_{\edits{\rm o}}=\mathbf{w}_{\edits{\rm o}}=\mu=1$, $\beta = \delta = 0$. 

\subsection*{Proof of Lemma \ref{lem:dominant-eig}}

i) If $\mathcal{G}$ is Class I, by Definition \ref{def:graph-classes} and Proposition \ref{prop:PerFr}.i,  $A$ is switching equivalent to a graph $\mathcal{G}'$  with adjacency matrix $A'$ that has the strong Perron-Frobenius property. If $\mathcal{G}$ is Class II, $A'$ is nonnegative and irreducible.  
Then by Proposition~\ref{prop:PerFr}.ii, $A'$ has a simple dominant eigenvalue, and since $A$ and $A'$ are similar by the transformation \eqref{eq:sim_isomorphism}, they are are co-spectral. ii)  The similarity transformation between $A$ and $A'$ \eqref{eq:sim_isomorphism} also relates their eigenvectors.

\subsection*{Proof of Theorem \ref{thm:stab}}
Following Lemma \ref{lem:eigs}, for any leading eigenvalue $ \eta(u,\lambda,\mu)$ of \eqref{eq:jac}, $ \operatorname{Re}(\eta(u,\lambda,\mu)) =  -d + u \xi $. When $u<u^*$, $\operatorname{Re}(\eta) < 0$ and therefore $\operatorname{Re}(\tilde{\eta}) < 0$ for all $\tilde{\eta} \in \sigma(J(\mathbf{0},u))$; for $u > u^*$, $\operatorname{Re}(\eta) > 0$. Stability conclusions follow by Lyapunov's indirect method \cite[Theorem 4.7]{Khalil2002}. The existence of an attracting center manifold follows by the Center Manifold Theorem \cite[Theorem 3.2.1]{guckenheimer2013nonlinear}. 
\subsection*{Proof of Lemma \ref{lem:graph-cond} }

 i)  By Lemma \ref{lem:dominant-eig},  $\lambda_{max} = \rho(A_{\edits{\rm a}})$ ($\mu_{max} = \rho(A_{\edits{\rm o}})$) is a simple dominant eigenvalue of $A_{\edits{\rm a}}$ ($A_{\edits{\rm o}}$), where $\operatorname{Re}(\lambda) < \lambda_{max}$ for any $\lambda \in \sigma(A_{\edits{\rm a}}) \setminus \{\lambda_{max}\}$ and $\operatorname{Re}(\mu) < \mu_{max}$ for any $\mu \in \sigma(A_{\edits{\rm o}}) \setminus \{ \mu_{max} \}$.  The rest follows by  definition of $\Lambda_1$. 

     ii) By Lemma \ref{lem:dominant-eig}, $|\lambda_i| < \lambda_{max}$ for any $\lambda_i \in \sigma(A_{\edits{\rm a}}) \setminus \{\lambda_{max} \}$ and $|\mu_j| < \mu_{max}$ for any $\mu_j \in \sigma(A_{\edits{\rm o}}) \setminus \{\mu_{max} \}$. Then $|\lambda_i \mu_j | = |\lambda_i| |\mu_j| < \lambda_{max} \mu_{max}$ unless $\lambda_i = \lambda_{max}$ and $\mu_{j} = \mu_{max}$. The proposition follows by the definition of $\Lambda_2$.
     
     iii) For undirected $\mathcal{G}_{\edits{\rm a}}$, $\mathcal{G}_{\edits{\rm o}}$, the matrices $A_{\edits{\rm a}}$, $A_{\edits{\rm o}}$ are symmetric with real-valued spectra. By Lemma \ref{lem:dominant-eig},  $\lambda_{max} = \rho(A_{\edits{\rm a}})$ ($\mu_{max} = \rho(A_{\edits{\rm o}})$) is a simple dominant eigenvalue of $A_{\edits{\rm a}}$ ($A_{\edits{\rm o}}$). Sign-symmetric graphs have symmetric spectra by \cite[Theorem 3.2]{belardo2019open}, which implies $\lambda_{min} = -\lambda_{max}$ is a simple eigenvalue of $A_{\edits{\rm a}}$ and $\mu_{min} = - \mu_{max}$ is a simple eigenvalue of $A_{\edits{\rm o}}$. Suppose $\lambda_i \in \sigma(A_{\edits{\rm a}})$, $\mu_j \in \sigma(A_{\edits{\rm o}})$. Then $|\lambda_i \mu_j| = |\lambda_i| |\mu_j| < \rho(A_{\edits{\rm a}}) \rho(A_{\edits{\rm o}})$ unless $\lambda_i = \pm \lambda_{max}$ and $\mu_j = \pm \mu_{max}$. The proposition follows by the definition of $\Lambda_2$.
     
     iv) For undirected $\mathcal{G}_{\edits{\rm a}}$, $\mathcal{G}_{\edits{\rm o}}$, the matrices $A_{\edits{\rm a}}$ and $A_{\edits{\rm o}}$ are symmetric and their spectra are real. By Lemma \ref{lem:dominant-eig},  $\lambda_{max} = \rho(A_{\edits{\rm a}})$ ($\mu_{max} = \rho(A_{\edits{\rm o}})$) is a simple dominant eigenvalue of $A_{\edits{\rm a}}$ ($A_{\edits{\rm o}}$). The spectra of $\sigma(A_{\edits{\rm a}})$, $\sigma(A_{\edits{\rm o}})$ are symmetric if and only if both matrices have characteristic polynomial with no even-degree terms \cite[Theorem 3.1]{li2023signed}. When this is satisfied, by the same argument as presented for statement iii), $\Lambda_2 = \{ (\lambda_{max}, \mu_{max}), (-\lambda_{max},-\mu_{max}) \}$. Suppose $\mathcal{G}_{\edits{\rm a}}$ does not satisfy the stated algebraic condition. Then $|\lambda_i| < \lambda_{max}$ for any $\lambda_{i} \in \sigma(A_{\edits{\rm a}}) \setminus \{\lambda_{max}\}$ and $|\lambda_i \mu_j | = |\lambda_i| |\mu_j| < \rho(A_{\edits{\rm a}}) \rho(A_{\edits{\rm o}})$ unless $\lambda_i = \lambda_{max}$ and $\mu_{j} = \mu_{max}$. An analogous argument holds if $\mathcal{G}_{\edits{\rm o}}$ does not satisfy the algebraic condition.

\subsection*{Proof of Proposition \ref{prop:param-effects-1}}

Let $Q(\lambda,\mu) = \gamma \operatorname{Re}(\lambda) + \beta \operatorname{Re}(\mu) + \delta \operatorname{Re}(\lambda \mu)$. 1) When $\delta = 0$, $\operatorname{max}_{(\lambda,\mu) \in \sigma(A_{\edits{\rm a}}) \times \sigma(A_{\edits{\rm o}})} Q(\lambda,\mu) = \gamma \max_{\lambda \in \sigma(A_{\edits{\rm a}})} \operatorname{Re}(\lambda) + \beta \max_{\mu \in \sigma(A_{\edits{\rm o}})} \operatorname{Re}(\mu)$. Any pair $(\lambda, \mu)$ satisfying this condition must be in the set $\Lambda_1$ by the definition of the set. By continuity of eigenvalues, this will hold for $\delta < \varepsilon$. 2) When $\beta = \gamma = 0$, $\operatorname{max}_{(\lambda,\mu) \in \sigma(A_{\edits{\rm a}}) \times \sigma(A_{\edits{\rm o}})} Q(\lambda,\mu) = \delta \operatorname{max}_{(\lambda,\mu) \in \sigma(A_{\edits{\rm a}}) \times \sigma(A_{\edits{\rm o}})} \operatorname{Re}(\lambda \mu)$ and the statement follows trivially by the definition of the set $\Lambda_2$. By continuity of eigenvalues, this will hold for $\gamma < \varepsilon_1$, $\beta < \varepsilon_2$.

\subsection*{Proof of Proposition \ref{prop:param-effects-2}}

1)
Suppose $\lambda,\Tilde{\lambda} \in \sigma(A_{\edits{\rm a}})$, $\mu,\Tilde{\mu} \in \sigma(A_{\edits{\rm o}})$,  $\operatorname{Re}(\lambda) = \lambda_{max}$, $\operatorname{Re}(\mu) = \mu_{max}$ and at least one of the following holds: $\operatorname{Re}(\Tilde{\lambda}) < \lambda_{max}$, $\operatorname{Re}(\Tilde{\mu}) < \mu_{max}$, i.e. $(\lambda,\mu) \in \Lambda_1$ and $(\Tilde{\lambda}, \Tilde{\mu}) \not \in \Lambda_1$. 
Then $\operatorname{Re}(\eta(u,\lambda,\mu)) - \operatorname{Re}(\eta(u,\tilde{\lambda},\tilde{\mu})) = u \gamma \big(\lambda_{max} - \operatorname{Re}(\tilde{\lambda})\big) + u \beta \big( \mu_{max} - \operatorname{Re}(\tilde{\mu}) \big) + u \delta \big( \operatorname{Re}(\lambda \mu) - \operatorname{Re}(\tilde{\lambda} \tilde{\mu})\big) > 0$ whenever $\gamma \big(\lambda_{max} - \operatorname{Re}(\tilde{\lambda}) \big) + \beta\big(\mu_{max} - \operatorname{Re}(\Tilde{\mu})\big) > - \delta \big( \operatorname{Re}(\lambda \mu) - \operatorname{Re}(\tilde{\lambda} \tilde{\mu})\big)$. 
Let $Q(\gamma,\beta,\Tilde{\lambda},\Tilde{\mu})$ be the left-hand side of this inequality and define $K := \min_{\lambda, \Tilde{\lambda}\in \sigma(A_{\edits{\rm a}}), \mu, \Tilde{\mu} \in \sigma(A_{\edits{\rm o}})}\big( \operatorname{Re}(\lambda \mu) - \operatorname{Re}(\tilde{\lambda} \tilde{\mu})\big) < 0$. 
Observe that $\frac{\partial Q }{ \partial \gamma} = \lambda_{max} - \operatorname{Re}(\tilde{\lambda}) \geq 0$ and $\frac{\partial Q }{ \partial \beta} = \mu_{max} - \operatorname{Re}(\Tilde{\mu}) \geq 0$, with at least one of the two inequalities being strict. 
Therefore increasing the magnitude of $\gamma$ and/or of $\beta$ will eventually make the inequality $Q(\gamma,\beta,\Tilde{\lambda},\Tilde{\mu}) > - \delta K := K_c > 0$ valid for any choice of $(\Tilde{\lambda}, \Tilde{\mu}) \in \sigma(A_{\edits{\rm a}})\times \sigma(A_{\edits{\rm o}})\setminus \Lambda_1$. Furthermore, suppose $(\hat{\lambda},\hat{\mu})\in \sigma(A_{\edits{\rm a}})\times \sigma(A_{\edits{\rm o}})$, $\operatorname{Re}(\hat{\lambda}) = \lambda_2$, and $\operatorname{Re}(\hat{\mu}) = \mu_2$. For any $(\Tilde{\lambda}, \Tilde{\mu}) \in \sigma(A_{\edits{\rm a}})\times \sigma(A_{\edits{\rm o}})\setminus \Lambda_1$, $Q(\gamma, \beta, \Tilde{\lambda}, \Tilde{\mu}) \geq Q(\gamma, \beta, \hat{\lambda}, \hat{\mu})$. Then the inequality $Q(\gamma, \beta, \hat{\lambda}, \hat{\mu}) > K_c$ implies that  that $\operatorname{Re}(\eta(u,\lambda,\mu)) >  \operatorname{Re}(\eta(u,\tilde{\lambda},\tilde{\mu}))$ for any $(\Tilde{\lambda}, \Tilde{\mu}) \in \sigma(A_{\edits{\rm a}})\times \sigma(A_{\edits{\rm o}})\setminus \Lambda_1$. This shows $(\lambda,\mu) \in \Lambda_1 \implies (\lambda,\mu) \in \Lambda$. 

Now suppose $(\lambda,\mu) \in \Lambda$ and $Q(\gamma,\beta,\hat{\lambda},\hat{\mu}) > K_c > 0$ for all $(\hat{\lambda},\hat{\mu}) \in \sigma(A_{\edits{\rm a}}) \times \sigma(A_{\edits{\rm o}})\setminus \Lambda_1$. 
Then by the definition of the set $\Lambda$,  $\operatorname{Re}(\eta(\lambda,\mu,u) ) - \operatorname{Re}(\eta(\Tilde{\lambda},\Tilde{\mu},u)) \geq 0$ for any $(\tilde{\lambda},\Tilde{\mu}) \in \sigma(A_{\edits{\rm a}}) \times \sigma(A_{\edits{\rm o}})$, with equality holding if and only if $(\tilde{\lambda},\Tilde{\mu}) \in \Lambda$. 
Suppose $\operatorname{Re}(\tilde{\lambda}) = \lambda_{max}$, $\operatorname{Re}(\Tilde{\mu}) = \mu_{max}$. 
Then
$\operatorname{Re}(\eta(u,\lambda,\mu)) - \operatorname{Re}(\eta(u,\tilde{\lambda},\tilde{\mu})) 
= u \gamma \big(\operatorname{Re}(\lambda) - \lambda_{max}\big) + u \beta \big( \operatorname{Re}(\mu) - \mu_{max} \big) + u \delta \big( \operatorname{Re}(\lambda \mu) - \operatorname{Re}(\tilde{\lambda} \tilde{\mu})\big) \geq 0$ which means $-Q(\gamma,\beta,\lambda,\mu) \geq -  \delta \big( \operatorname{Re}(\lambda \mu) - \operatorname{Re}(\tilde{\lambda} \tilde{\mu})\big)$. 
If $(\lambda, \mu) \not \in \Lambda_1$, then $-K_c > - Q(\gamma,\beta,\lambda,\mu) > -\delta \big( \operatorname{Re}(\lambda \mu) - \operatorname{Re}(\tilde{\lambda} \tilde{\mu})\big) \implies K_c < -\delta \big(\operatorname{Re}(\tilde{\lambda} \tilde{\mu}) - \operatorname{Re}(\lambda \mu) \big)$. This is a contradiction of the definition $K_c = \operatorname{max}_{\lambda, \Tilde{\lambda}\in \sigma(A_{\edits{\rm a}}), \mu, \Tilde{\mu} \in \sigma(A_{\edits{\rm o}})} -\delta \big(\operatorname{Re}(\lambda \mu) - \operatorname{Re}(\Tilde{\lambda} \Tilde{\mu}) \big)$. Therefore $(\lambda,\mu) \in \Lambda_1$ and the proposition follows. 
}

2) Suppose $\lambda,\Tilde{\lambda} \in \sigma(A_{\edits{\rm a}})$, $\mu,\Tilde{\mu} \in \sigma(A_{\edits{\rm o}})$,$\operatorname{Re}(\lambda \mu) = (\lambda \mu)_{max}$, $\operatorname{Re}(\Tilde{\lambda} \Tilde{\mu}) < (\lambda \mu)_{max}$, i.e. $(\lambda,\mu) \in \Lambda_2$ and $(\Tilde{\lambda},\Tilde{\mu}) \not \in \Lambda_2$. 
Then $\operatorname{Re}(\eta(u,\lambda,\mu)) - \operatorname{Re}(\eta(u,\tilde{\lambda},\tilde{\mu})) 
= u \gamma \big(\operatorname{Re}(\lambda) - \operatorname{Re}(\tilde{\lambda})\big) + u \beta \big( \operatorname{Re}(\mu) - \operatorname{Re}(\tilde{\mu}) \big) + u \delta \big( (\lambda \mu)_{max} - \operatorname{Re}(\tilde{\lambda} \tilde{\mu})\big) > 0$
whenever 
\begin{equation}
    \delta > \frac{ - \beta \big( \operatorname{Re}(\mu) - \operatorname{Re}(\tilde{\mu}) \big) - \gamma \big( \operatorname{Re}(\lambda) - \operatorname{Re}(\tilde{\lambda})\big) }{(\lambda \mu)_{max} - \operatorname{Re}(\Tilde{\lambda} \Tilde{\mu})}. \label{eq:lambd_cond}
\end{equation}
Define the right-hand side of \eqref{eq:lambd_cond} to be $h(\lambda,\mu,\Tilde{\lambda},\Tilde{\mu})$ and let $\delta_c := \max_{\lambda, \Tilde{\lambda}\in \sigma(A_{\edits{\rm a}}), \mu, \Tilde{\mu} \in \sigma(A_{\edits{\rm o}}), (\Tilde{\lambda},\tilde{\mu}) \not \in \Lambda_2} h(\lambda,\mu,\Tilde{\lambda},\Tilde{\mu})$. 
Then assuming $\delta > \delta_c$, $\operatorname{Re}(\eta(u,\lambda,\mu)) > \operatorname{Re}(\eta(u,\Tilde{\lambda},\Tilde{\mu}))$ for any $(\Tilde{\lambda},\Tilde{\mu}) \not \in \Lambda_2$, which means $(\lambda,\mu) \in \Lambda_1$ implies $(\lambda,\mu) \in \Lambda$.

$\operatorname{Re}(\lambda) = \lambda_{max}$ implies that $\operatorname{Re}(\eta(u,\lambda,\mu)) >  \operatorname{Re}(\eta(u,\tilde{\lambda},\tilde{\mu}))$ for any $\Tilde{\lambda} \in \sigma(A_{\edits{\rm a}})$ with $\operatorname{Re}(\Tilde{\lambda}) < \lambda_{max}$ and any choice of $\mu, \Tilde{\mu} \in \sigma(A_{\edits{\rm o}})$. 

Next, suppose $\delta > \delta_c$, $(\lambda,\mu) \in \Lambda$, and $\operatorname{Re}(\lambda \mu) < (\lambda \mu)_{max}$. 
Then since $\delta > \delta_c$, for any eigenvalue pair $\hat{\lambda}\in \sigma(A_{\edits{\rm a}})$, $\hat{\mu} \in \sigma(A_{\edits{\rm o}})$ with $\operatorname{Re}(\hat{\lambda} \hat{\mu}) = (\lambda \mu)_{max}$ it must hold that $ \operatorname{Re}(\eta(\hat{\lambda},\hat{\mu},u)) > \operatorname{Re}(\eta,\lambda, u)$. 
However, by  definition of  set $\Lambda$ \eqref{eq:Lambda}, it must hold that $\operatorname{Re}(\eta(\lambda,\mu,u)) \geq \operatorname{Re}(\eta(\hat{\lambda},\hat{\mu},u))$, which is a contradiction.
Therefore, whenever $\gamma > \gamma_c$ and $(\lambda,\mu) \in \Lambda$, it must hold that $\operatorname{Re}(\lambda \mu) = (\lambda \mu)_{max}$, i.e. $(\lambda,\mu) \in \Lambda_2$.



\subsection*{Proof of Theorem \ref{thm:pitchfork-No}}

To identify the pitchfork bifurcation in the model \eqref{EQ:value_dynamics} we rely on the singularity theory of bifurcations \cite{Golubitsky1985}. 
For a concise overview of the main ideas of this approach we refer the reader to Chapter 2.3 of the doctoral dissertation \cite{AnastasiaThesis}.

1) From Theorem \ref{thm:stab} we know $\Zz = \mathbf{0}$ is stable for $u < u^*$ and unstable for $u > u^*$; when $u = u^*$, a bifurcation happens along a manifold tangent to $\operatorname{span}\{\mathbf{v}_{\edits{\rm a}} \otimes \mathbf{v}_{\edits{\rm o}}\}$. To identify this bifurcation we compute a polynomial expansion of a Lyapunov-Schmidt (LS) reduction $f(y,u,\mathbf{b})$ of \eqref{EQ:value_dynamics} following the procedure outlined in \cite[Chapter I]{Golubitsky1985}. We derive the coefficients through third order. With the odd symmetry assumption on $S_m$, $S_m''(0) = 0$ and a LS reduction approximation reads
$f(y,u,\mathbf{0}) = K_0 (u - u^*) y  + d K_1 / K_0  \ y^3$. 
To establish a pitchfork bifurcation, we verify that this LS reduction satisfies the necessary and sufficient conditions outlined in \cite[Chapter II, Proposition 9.2]{Golubitsky1985}: $f(0,u^*,\mathbf{0}) = \frac{\partial f}{\partial y}(0,u^*,\mathbf{0}) = \frac{\partial^2 f}{\partial y^2}(0,u^*,\mathbf{0}) = \frac{\partial f}{\partial u}(0,u^*,\mathbf{0}) = 0$, $\frac{\partial^2 f}{ \partial y \partial u}(0,u^*,\mathbf{0}) \neq 0$, $\frac{\partial^3 f}{ \partial y^3}(0,u^*,\mathbf{0}) \neq 0$. Linearization of the system about each of the two bifurcating fixed points shares the $N-1$ stable eigenvalues with the origin, and the bifurcating eigenvalue is negative (or positive) under the conditions in the theorem statement by \cite[Chapter I, Theorem 4.1]{Golubitsky1985}.

2) and 3) With nonzero $\mathbf{b}$, $g_1$, $g_2$, the coefficients of the polynomial expansion of $f(y,u,\mathbf{b})$ through third order take the form $ f(y,u,\mathbf{b}) = K_0 (u - u^*) y  +  K_b + K_3 y^3 + d K_2/K_0 y^2$
where $K_b = \langle \mathbf{w}_{\edits{\rm a}} \otimes \mathbf{w}_{\edits{\rm o}}, \mathbf{b} \rangle $ and $K_3 = d K_1/K_0 - 3 \langle \mathbf{w}_{\edits{\rm a}} \otimes \mathbf{w}_{\edits{\rm o}}, (u^*)^2 (g_1''(0) (\alpha +\gamma \lambda)^2 +  g_2''(0) \mu^{2} (\beta + \delta \lambda)^2 )
  (g_1''(0) (\alpha +\gamma \lambda) \mathbf{v}_{\edits{\rm a}} \otimes \mathbf{v}_{\edits{\rm o}} \odot ( (\alpha \mathcal{I}_{N_{\edits{\rm a}}} + \gamma A_{\edits{\rm a}}) \otimes \mathcal{I}_{N_{\edits{\rm o}}} ) (J^{-1} P (\mathbf{v}_{\edits{\rm a}} \otimes \mathbf{v}_{\edits{\rm o}})^2) + g_2''(0)\mu (\beta +\delta \lambda) \mathbf{v}_{\edits{\rm a}} \otimes \mathbf{v}_{\edits{\rm o}} \odot \left( (\beta \mathcal{I}_{N_{\edits{\rm a}}} + \delta A_{\edits{\rm a}}) \otimes A_{\edits{\rm o}} \right) (J^{-1} P (\mathbf{v}_{\edits{\rm a}} \otimes \mathbf{v}_{\edits{\rm o}})^2) ) \rangle$  where $J^{-1}$ is the inverse of the restriction of $J(\mathbf{0},u^*)$ to its range, and $P = \mathcal{I}_{N_{\edits{\rm a}} N_{\edits{\rm o}}} - \frac{1}{\| \mathbf{v}_{\edits{\rm a}} \otimes \mathbf{v}_{\edits{\rm o}} \|^2} (\mathbf{v}_{\edits{\rm a}} \otimes \mathbf{v}_{\edits{\rm o}})(\mathbf{v}_{\edits{\rm a}} \otimes \mathbf{v}_{\edits{\rm o}})^T$ is a projection onto the range of $J(\mathbf{0},u^*)$. Observe that whenever $g_1''(0) = 0$ and $g_2''(0) = 0$, the cubic coefficient $K_3 = d K_1/K_0$. Therefore by continuity there exist small positive $\varepsilon_{g1}$, $\varepsilon_{g2}$ such that $\operatorname{sign}(K_3) = \operatorname{sign}(K_1)$ whenever $|g_1''(0)| < \varepsilon_{g1}$ and $|g_2''(0)| < \varepsilon_{g2}$. Observe that $K_2,K_b$ are small. The bifurcation diagram $f(y,u,\mathbf{b}) = 0$ is an unfolding (i.e. a perturbation) of a symmetric pitchfork bifurcation diagram, and the characterization of equilibria as a function of $K_b, K_2$ follow from standard analysis of a pitchfork bifurcation universal unfolding \cite[Chapter I \S 1]{Golubitsky1985}.

\subsection*{Proof of Corollary \ref{cor:pitchfork_conds}}
 i) By Lemma~\ref{lem:graph-cond}.i, $\Lambda_1 = \{(\lambda_{max},\mu_{max})\}$. Since $\lambda_{max},\mu_{max} > 0$, $\alpha + \gamma \lambda_{max} + \beta \mu_{max} + \delta \lambda_{max} \mu_{max} > 0$ and the conditions of Theorem \ref{thm:pitchfork-No} are satisfied whenever $\Lambda = \Lambda_1$. 
    ii) By Lemma \ref{lem:graph-cond}.ii, $\Lambda = \Lambda_1$ for any choice of model parameters. By the argument in part i), the conditions of Theorem \ref{thm:pitchfork-No} are satisfied and the statement follows.     
    iii) By Theorem \ref{thm:bif}, the equilibria lie on a manifold that is approximated arbitrarily closely by the span of $\mathbf{v}_{\edits{\rm a}} \otimes \mathbf{v}_{\edits{\rm o}}$ when $|u - u^*|$ is small. By Lemma \ref{lem:dominant-eig}.ii, there exists a vector $\mathbf{v}' \succ \mathbf{0}$ such that $\mathbf{v}_{\edits{\rm a}} = \Theta P_n \mathbf{v}'$. Vector $P_n \mathbf{v}' \succ \mathbf{0}$. Thus, $(\mathbf{v}_{\edits{\rm a}})_i > 0 (< 0)$ if and only if $\theta(i) = 1 (-1)$. Belief $z_{ij}^* \approx (\mathbf{v}_{\edits{\rm a}})_i (\mathbf{v}_{\edits{\rm o}})_j$; as long as $(\mathbf{v}_{\edits{\rm o}})_j \neq 0$, $\operatorname{sign}(z_{kj}^*) = \operatorname{sign}((\mathbf{v}_{\edits{\rm a}})_k (\mathbf{v}_{\edits{\rm o}})_j) = \operatorname{sign}((\mathbf{v}_{\edits{\rm a}})_i (\mathbf{v}_{\edits{\rm o}})_j) = \operatorname{sign}(z_{ij}^*)$ if and only if $\theta(i) = \theta(k)$.
    iv) The statement follows by analogy to part iii).

\subsection*{Proof of Theorem \ref{thm:hopf}}


This proof is mostly verbatim from \cite[Theorem IV.3]{bizyaeva2022sustained}, with the argument generalized slightly to accommodate asymmetric $S_1,S_2$. Let $C = \alpha + \gamma \operatorname{Re}(\lambda^\dagger) + \beta \operatorname{Re}(\mu^\dagger) + \delta \operatorname{Re}(\lambda^\dagger \mu^\dagger)$.

1) To establish existence of periodic orbits we check that the system \eqref{EQ:value_dynamics} under the stated assumptions satisfies the conditions of the Hopf bifurcation theorem \cite[Theorem 3.4.2]{guckenheimer2013nonlinear}. When $u = u^* = d/C$, the leading eigenvalues of \eqref{eq:jac} are a simple purely imaginary pair
$
\eta_{\pm}(u^*) = \pm i u^* \big| \gamma \operatorname{Im}(\lambda^\dagger) + \beta \operatorname{Im}(\mu^{\dagger}) + \delta \operatorname{Im}(\lambda^{\dagger}\mu^{\dagger}) \big| \neq 0 
$, 
which satisfies the eigenvalue condition (H1) of the Hopf theorem. Next, we check that the leading eigenvalues cross the imaginary axis with nonzero speed as $u$ is varied, i.e. 
$
    \frac{d}{du}\operatorname{Re}(\eta_{\pm}(u)) =  C > 0 
$, 
which satisfies the nonzero crossing speed condition (H2) of the Hopf theorem. Existence of periodic orbits directly follows by the Hopf theorem. By this theorem and by the definition of a center manifold \cite[Theorem 3.2.1]{guckenheimer2013nonlinear}, the solutions appear along a unique $W^s$ which is tangent at $u = u^*$ to $\operatorname{Ker}\big(J(\mathbf{0},u^*)\big) = \operatorname{span}\{ \operatorname{Re}(\mathbf{v}_{\edits{\rm a}} \otimes \mathbf{v}_{\edits{\rm o}}), \operatorname{Im}(\mathbf{v}_{\edits{\rm a}} \otimes \mathbf{v}_{\edits{\rm o}}) \} $.

To show 2) and 3) we first compute 
a third-order approximation of \eqref{EQ:value_dynamics} following the Lyapunov-Schmidt reduction for a Hopf bifurcation  \cite[Chapter VIII, Proposition 3.3]{Golubitsky1985}:  
$
    f(y,u) = C y (u - u^*)  + \frac{1}{16} u^* K_3 y^3
$, where $K_3 = K + K_g(g_1''(0), g_2''(0))$ and $K_g(0,0) = 0$.
As long as $K_3 \neq 0$, by \cite[Chapter VIII, Theorems 2.1 and 3.2]{Golubitsky1985} the reduced bifurcation equation $f(y,u)$ possesses a pitchfork bifurcation which is supercritical for $K_3<0$ and subcritical for $K_3>0$. By continuity there exist small positive $\varepsilon_{g1},\varepsilon_{g2}$ such that $\operatorname{sign}(K_3) = \operatorname{sign}(K)$ whenever $|g_1''(0)| < \varepsilon_{g1}$ and $|g_2''(0)| < \varepsilon_{g2}$ from which the theorem statement follows.

When
$\lvert u - u^*\rvert$ is small, solutions to $f(y,u) = 0$ are in one-to-one correspondence with orbits of small amplitude periodic solutions to the system \eqref{EQ:value_dynamics} with period near $ 2 \pi/(u^*|\gamma \operatorname{Im}(\lambda^\dagger) + \beta \operatorname{Im}(\mu^{\dagger}) + \delta \operatorname{Im}(\lambda^{\dagger}\mu^{\dagger})|):=1/\omega$. For $u$ near $u^*$, the small amplitude oscillations can be approximated to first order as scalar multiples of $e^{i \omega t} \mathbf{v}_{\edits{\rm a}} \otimes \mathbf{v}_{\edits{\rm o}}$ from which the conclusions on phase and amplitude difference between agents follow. When $K_3 < 0$ ($>0$), and therefore by the same continuity argument as above when $K < 0$ ($>0$), the bifurcating periodic solutions are stable (unstable) as established in \cite[Chapter VIII, Theorem 4.1]{Golubitsky1985}. 

\subsection*{Proof of Corollary \ref{cor:balanced-oscillations}}

i) The set $\Lambda_1 = \{(\lambda_{max}, \mu), (\lambda_{max},\overline{\mu}) \}$ corresponding to $\mathcal{G}_{\edits{\rm a}}$, $\mathcal{G}_{\edits{\rm o}}$ satisfies the first case for the necessary condition in Assumption \ref{ass1} and any social dissonance-driven bifurcation is a Hopf bifurcation. By part 3) of Theorem \ref{ex:hopf}, phase difference between $z_{ij}^*(t)$ and $z_{kj}^*(t)$ is near $\varphi_{ik}^{jj} = \operatorname{arg}((\mathbf{v}_{\edits{\rm a}})_i (\mathbf{v}_{\edits{\rm o}})_j) - \operatorname{arg}((\mathbf{v}_{\edits{\rm a}})_k (\mathbf{v}_{\edits{\rm o}})_j) = \operatorname{arg}((\mathbf{v}_{\edits{\rm a}})_i) + \operatorname{arg}(\mathbf{v}_{\edits{\rm o}})_j) - \operatorname{arg}((\mathbf{v}_{\edits{\rm a}})_k) - \operatorname{arg}(\mathbf{v}_{\edits{\rm o}})_j) = \operatorname{arg}\left( \frac{(\mathbf{v}_{\edits{\rm a}})_i}{(\mathbf{v}_{\edits{\rm a}})_k} \right)$ which means, since all entries of $\mathbf{v}_{\edits{\rm a}}$ are real-valued, $\varphi_{ik}^{jj} = 0$ if $\operatorname{sign}((\mathbf{v}_{\edits{\rm a}})_{i}) = \operatorname{sign}((\mathbf{v}_{\edits{\rm a}})_{k})$ and $\varphi_{ik}^{jj} = \pi$ if $\operatorname{sign}((\mathbf{v}_{\edits{\rm a}})_{i}) \neq \operatorname{sign}((\mathbf{v}_{\edits{\rm a}})_{k})$. By Lemma \ref{lem:dominant-eig}.ii, there exists a vector $\mathbf{v}' \succ \mathbf{0}$ such that $\mathbf{v}_{\edits{\rm a}} = \Theta P_n \mathbf{v}'$. 
$P_n \mathbf{v}' \succ \mathbf{0}$, and it follows that $(\mathbf{v}_{\edits{\rm a}})_i > 0 (< 0)$ if and only if $\theta(i) = 1 (-1)$. So $\varphi_{ik}^{jj} = 0$ whenever $i,k$ both belong to either $\mathcal{V}_{\edits{\rm a 1}}$ or $\mathcal{V}_{\edits{\rm a 2}}$; otherwise $\varphi_{ik}^{jj} = \pi$.  Part ii) follows by an analogous argument.

\subsection*{Proof of Corollary \ref{cor:graph-cond-Hopf}}

1) Due to the symmetry of the adjacency matrices, eigenvalues $\lambda \in \sigma(A_{\edits{\rm a}})$ and $\mu \in \sigma(A_{\edits{\rm o}})$ must be real-valued, which violates the necessary conditions in Assumption \ref{ass1}. 

    2) Without loss of generality, let $\mathcal{G}_{\edits{\rm o}}$ be the circulant graph. 
    By Lemma \ref{lem:dominant-eig}, $A_{\edits{\rm a}}$ has a simple dominant eigenvalue $\lambda$. $A$ has $N = N_{\edits{\rm o}}$ eigenvalues $\mu_m =  - \operatorname{exp}(2 \pi m i/N)$, $m = 0, \dots, N-1$ \cite[2.2.P10]{horn2012matrix}. A leading eigenvalue of $A$, and thus of $A_{\edits{\rm o}}$, since the two are co-spectral, has real part $\lambda_{max} = \max_{m = 0, \dots, N-1} \operatorname{Re}(\mu_m) =\max_{m = 0, \dots, N-1}  - \cos(2 \pi m/N)$, i.e., at $m = (N-1)/2$ and $m = (N+1)/2$ since $N$ is odd. Since $\mu_{(N-1)/2} = \overline{\mu}_{(N+1)/2}$, the set of leading eigenvalue pairs is  $\Lambda_1 = \{ (\lambda, \mu_{(N+1)/2}),(\lambda,\overline{\mu}_{(N+1)/2})\}$ which satisfies the necessary condition in Assumption \ref{ass1}.

\bibliographystyle{IEEEtran}
\bibliography{references}

\begin{thebibliography}{10}
\providecommand{\url}[1]{#1}
\csname url@samestyle\endcsname
\providecommand{\newblock}{\relax}
\providecommand{\bibinfo}[2]{#2}
\providecommand{\BIBentrySTDinterwordspacing}{\spaceskip=0pt\relax}
\providecommand{\BIBentryALTinterwordstretchfactor}{4}
\providecommand{\BIBentryALTinterwordspacing}{\spaceskip=\fontdimen2\font plus
\BIBentryALTinterwordstretchfactor\fontdimen3\font minus \fontdimen4\font\relax}
\providecommand{\BIBforeignlanguage}[2]{{%
\expandafter\ifx\csname l@#1\endcsname\relax
\typeout{** WARNING: IEEEtran.bst: No hyphenation pattern has been}%
\typeout{** loaded for the language `#1'. Using the pattern for}%
\typeout{** the default language instead.}%
\else
\language=\csname l@#1\endcsname
\fi
#2}}
\providecommand{\BIBdecl}{\relax}
\BIBdecl

\bibitem{olfati2004consensus}
R.~Olfati-Saber and R.~M. Murray, ``Consensus problems in networks of agents with switching topology and time-delays,'' \emph{IEEE Trans. Autom. Control}, vol.~49, no.~9, pp. 1520--1533, 2004.

\bibitem{montes2010opinion}
M.~A. Montes~de Oca, E.~Ferrante, N.~Mathews, M.~Birattari, and M.~Dorigo, ``Opinion dynamics for decentralized decision-making in a robot swarm,'' in \emph{Swarm Intelligence. ANTS 2010.}, vol. 6234.\hskip 1em plus 0.5em minus 0.4em\relax Springer, 2010, pp. 251--262.

\bibitem{montes2011majority}
M.~A. Montes~de Oca, E.~Ferrante, A.~Scheidler, C.~Pinciroli, M.~Birattari, and M.~Dorigo, ``Majority-rule opinion dynamics with differential latency: a mechanism for self-organized collective decision-making,'' \emph{Swarm Intelligence}, vol.~5, pp. 305--327, 2011.

\bibitem{altafini2012consensus}
C.~Altafini, ``Consensus problems on networks with antagonistic interactions,'' \emph{IEEE Trans. Autom. Control}, vol.~58, no.~4, pp. 935--946, 2012.

\bibitem{HaiminCDC2023}
H.~Hu, K.~Nakamura, K.-C. Hsu, N.~E. Leonard, and J.~F. Fisac, ``Emergent coordination through game-induced nonlinear opinion dynamics,'' \emph{Proc. IEEE CDC}, pp. 8122--8129, 2023.

\bibitem{cathcart2022opinion}
C.~Cathcart, M.~Santos, S.~Park, and N.~E. Leonard, ``Proactive opinion-driven robot navigation around human movers,'' \emph{Proc. IROS}, pp. 4052--4058, 2023.

\bibitem{leung2023leveraging}
H.~C. Leung, Z.~Li, B.~She, and P.~E. Par{\'e}, ``Leveraging opinions and vaccination to eradicate networked epidemics,'' in \emph{2023 European Control Conference}, 2023, pp. 1--6.

\bibitem{paine2023model}
T.~M. Paine and M.~R. Benjamin, ``\edits{A Model for Multi-Agent Autonomy That Uses Opinion Dynamics and Multi-Objective Behavior Optimization},'' \emph{Proc. IEEE ICRA}, 2024.

\bibitem{galesic2021human}
M.~Galesic, W.~Bruine~de Bruin, J.~Dalege, S.~L. Feld, F.~Kreuter, H.~Olsson, D.~Prelec, D.~L. Stein, and T.~van Der~Does, ``Human social sensing is an untapped resource for computational social science,'' \emph{Nature}, vol. 595, no. 7866, pp. 214--222, 2021.

\bibitem{dalege2023networks}
J.~Dalege, M.~Galesic, and H.~Olsson, ``Networks of beliefs: An integrative theory of individual-and social-level belief dynamics,'' \emph{OSF Preprints}, 2023.

\bibitem{vlasceanu2022network}
M.~Vlasceanu, A.~M. Dyckovsky, and A.~Coman, ``A network approach to investigate the dynamics of individual and collective beliefs: Advances and applications of the bending model,'' \emph{Perspect. Psychol. Sci.}, 2023.

\bibitem{suzuki2015neural}
S.~Suzuki, R.~Adachi, S.~Dunne, P.~Bossaerts, and J.~P. O’Doherty, ``Neural mechanisms underlying human consensus decision-making,'' \emph{Neuron}, vol.~86, no.~2, pp. 591--602, 2015.

\bibitem{strandburg2013visual}
A.~Strandburg-Peshkin, C.~R. Twomey, N.~W. Bode, A.~B. Kao, Y.~Katz, C.~C. Ioannou, S.~B. Rosenthal, C.~J. Torney, H.~S. Wu, S.~A. Levin \emph{et~al.}, ``Visual sensory networks and effective information transfer in animal groups,'' \emph{Current Biology}, vol.~23, no.~17, pp. R709--R711, 2013.

\bibitem{rosenthal2015revealing}
S.~B. Rosenthal, C.~R. Twomey, A.~T. Hartnett, H.~S. Wu, and I.~D. Couzin, ``Revealing the hidden networks of interaction in mobile animal groups allows prediction of complex behavioral contagion,'' \emph{Proc. Natl. Acad. Sci. U.S.A.}, vol. 112, no.~15, pp. 4690--4695, 2015.

\bibitem{sridhar2021geometry}
V.~H. Sridhar, L.~Li, D.~Gorbonos, M.~Nagy, B.~R. Schell, T.~Sorochkin, N.~S. Gov, and I.~D. Couzin, ``The geometry of decision-making in individuals and collectives,'' \emph{Proc. Natl. Acad. Sci. U.S.A.}, vol. 118, no.~50, 2021.

\bibitem{oscar2023simple}
L.~Oscar, L.~Li, D.~Gorbonos, I.~Couzin, and N.~Gov, ``A simple cognitive model explains movement decisions during schooling in zebrafish,'' \emph{bioRxiv}, pp. 2023--02, 2023.

\bibitem{bizyaeva2022}
A.~Bizyaeva, A.~Franci, and N.~E. Leonard, ``Nonlinear opinion dynamics with tunable sensitivity,'' \emph{IEEE Trans. Autom. Control}, vol.~68, no.~3, pp. 1415--1430, 2022.

\bibitem{franci2022breaking}
A.~Franci, M.~Golubitsky, I.~Stewart, A.~Bizyaeva, and N.~E. Leonard, ``Breaking indecision in multiagent, multioption dynamics,'' \emph{SIAM J. Appl. Dynamical Syst.}, vol.~22, no.~3, pp. 1780--1817, 2023.

\bibitem{degroot1974reaching}
M.~H. DeGroot, ``Reaching a consensus,'' \emph{Journal of the American Statistical Association}, vol.~69, no. 345, pp. 118--121, 1974.

\bibitem{friedkin2016network}
N.~E. Friedkin, A.~V. Proskurnikov, R.~Tempo, and S.~E. Parsegov, ``Network science on belief system dynamics under logic constraints,'' \emph{Science}, vol. 354, no. 6310, pp. 321--326, 2016.

\bibitem{parsegov2016novel}
S.~E. Parsegov, A.~V. Proskurnikov, R.~Tempo, and N.~E. Friedkin, ``Novel multidimensional models of opinion dynamics in social networks,'' \emph{IEEE Trans. Autom. Control}, vol.~62, no.~5, pp. 2270--2285, 2016.

\bibitem{ye2019consensus}
M.~Ye, J.~Liu, L.~Wang, B.~D. Anderson, and M.~Cao, ``Consensus and disagreement of heterogeneous belief systems in influence networks,'' \emph{IEEE Trans. Autom. Control}, vol.~65, no.~11, pp. 4679--4694, 2019.

\bibitem{he2022opinion}
G.~He, Z.~Ci, X.~Wu, and M.~Hu, ``Opinion dynamics with antagonistic relationship and multiple interdependent topics,'' \emph{IEEE Access}, vol.~10, pp. 31\,595--31\,606, 2022.

\bibitem{pan2018bipartite}
L.~Pan, H.~Shao, M.~Mesbahi, Y.~Xi, and D.~Li, ``Bipartite consensus on matrix-valued weighted networks,'' \emph{IEEE Trans. Circuits Syst. II}, vol.~66, no.~8, pp. 1441--1445, 2018.

\bibitem{nedic2019graph}
A.~Nedi{\'c}, A.~Olshevsky, and C.~A. Uribe, ``Graph-theoretic analysis of belief system dynamics under logic constraints,'' \emph{Sci. Rep.}, vol.~9, no.~1, p. 8843, 2019.

\bibitem{ye2020continuous}
M.~Ye, M.~H. Trinh, Y.-H. Lim, B.~D. Anderson, and H.-S. Ahn, ``Continuous-time opinion dynamics on multiple interdependent topics,'' \emph{Automatica}, vol. 115, p. 108884, 2020.

\bibitem{ahn2020opinion}
H.-S. Ahn, Q.~Van~Tran, M.~H. Trinh, M.~Ye, J.~Liu, and K.~L. Moore, ``Opinion dynamics with cross-coupling topics: Modeling and analysis,'' \emph{IEEE Trans. Comput. Soc. Syst.}, vol.~7, no.~3, pp. 632--647, 2020.

\bibitem{Liu2021interplay}
F.~Liu, S.~Cui, W.~Mei, F.~Dörfler, and M.~Buss, ``Interplay between homophily-based appraisal dynamics and influence-based opinion dynamics: Modeling and analysis,'' \emph{IEEE Control Systems Letters}, vol.~5, no.~1, pp. 181--186, 2021.

\bibitem{wang2022characterizing}
C.~Wang, L.~Pan, H.~Shao, D.~Li, and Y.~Xi, ``Characterizing bipartite consensus on signed matrix-weighted networks via balancing set,'' \emph{Automatica}, vol. 141, p. 110237, 2022.

\bibitem{he2023opinion}
G.~He, Z.~Shen, T.~Huang, W.~Zhang, and X.~Wu, ``Opinion dynamics with heterogeneous multiple interdependent topics on the signed social networks,'' \emph{IEEE Trans. Syst. Man. Cybern}, 2023.

\bibitem{lorenz2007continuous}
J.~Lorenz, ``Continuous opinion dynamics under bounded confidence: A survey,'' \emph{Int. J. Mod. Phys. C}, vol.~18, no.~12, pp. 1819--1838, 2007.

\bibitem{dandekar2013biased}
P.~Dandekar, A.~Goel, and D.~T. Lee, ``Biased assimilation, homophily, and the dynamics of polarization,'' \emph{Proc. Natl. Acad. Sci. U.S.A.}, vol. 110, no.~15, pp. 5791--5796, 2013.

\bibitem{Gray2018}
R.~{Gray}, A.~{Franci}, V.~{Srivastava}, and N.~E. {Leonard}, ``Multiagent decision-making dynamics inspired by honeybees,'' \emph{IEEE Trans. Control Netw. Syst.,}, vol.~5, no.~2, pp. 793--806, 2018.

\bibitem{xia2020analysis}
W.~Xia, M.~Ye, J.~Liu, M.~Cao, and X.-M. Sun, ``Analysis of a nonlinear opinion dynamics model with biased assimilation,'' \emph{Automatica}, vol. 120, p. 109113, 2020.

\bibitem{Wang2022}
L.~Wang, Y.~Hong, G.~Shi, and C.~Altafini, ``Signed social networks with biased assimilation,'' \emph{IEEE Trans. Autom. Control}, vol.~67, no.~10, pp. 5134--5149, 2022.

\bibitem{mei2022micro}
W.~Mei, F.~Bullo, G.~Chen, J.~M. Hendrickx, and F.~D{\"o}rfler, ``Micro-foundation of opinion dynamics: Rich consequences of the weighted-median mechanism,'' \emph{Phys. Rev. Res.}, vol.~4, no.~2, 2022.

\bibitem{shrinate2022nonlinear}
A.~Shrinate, T.~Tripathy, and L.~Behera, ``Nonlinear opinion dynamics using disagreement laplacian flows in antagonistic networks,'' in \emph{2022 Eighth Indian Control Conference (ICC)}.\hskip 1em plus 0.5em minus 0.4em\relax IEEE, 2022, pp. 278--283.

\bibitem{UM2001}
M.~Usher and J.~L. McClelland, ``The time course of perceptual choice: The leaky, competing accumulator model,'' \emph{Psychol. Rev.}, vol. 108, no.~3, pp. 550--592, 2001.

\bibitem{bogacz2007extending}
R.~Bogacz, M.~Usher, J.~Zhang, and J.~L. McClelland, ``Extending a biologically inspired model of choice: multi-alternatives, nonlinearity and value-based multidimensional choice,'' \emph{Philos. Trans. R. Soc. Lond., B, Biol. Sci.}, vol. 362, no. 1485, pp. 1655--1670, 2007.

\bibitem{introne2023measuring}
J.~Introne, ``Measuring belief dynamics on twitter,'' in \emph{Proc. of the Int. AAAI Conf. on Web and Social Media}, vol.~17, 2023, pp. 387--398.

\bibitem{rodriguez2016collective}
N.~Rodriguez, J.~Bollen, and Y.-Y. Ahn, ``Collective dynamics of belief evolution under cognitive coherence and social conformity,'' \emph{PLoS one}, vol.~11, no.~11, 2016.

\bibitem{baumann2020modeling}
F.~Baumann, P.~Lorenz-Spreen, I.~M. Sokolov, and M.~Starnini, ``Modeling echo chambers and polarization dynamics in social networks,'' \emph{Phys. Rev. Lett.}, vol. 124, no.~4, 2020.

\bibitem{siegenfeld2020negative}
A.~F. Siegenfeld and Y.~Bar-Yam, ``Negative representation and instability in democratic elections,'' \emph{Nat. Phys.}, vol.~16, no.~2, pp. 186--190, 2020.

\bibitem{baumann2021emergence}
F.~Baumann, P.~Lorenz-Spreen, I.~M. Sokolov, and M.~Starnini, ``Emergence of polarized ideological opinions in multidimensional topic spaces,'' \emph{Physical Review X}, vol.~11, no.~1, p. 011012, 2021.

\bibitem{brandt2021evaluating}
M.~J. Brandt and W.~W. Sleegers, ``Evaluating belief system networks as a theory of political belief system dynamics,'' \emph{Personality and Social Psychology Review}, vol.~25, no.~2, pp. 159--185, 2021.

\bibitem{gajewski2022transitions}
{\L}.~G. Gajewski, J.~Sienkiewicz, and J.~A. Ho{\l}yst, ``Transitions between polarization and radicalization in a temporal bilayer echo-chamber model,'' \emph{Phy. Rev. E}, vol. 105, no.~2, p. 024125, 2022.

\bibitem{li2022modeling}
L.~Li, A.~Zeng, Y.~Fan, and Z.~Di, ``Modeling multi-opinion propagation in complex systems with heterogeneous relationships via {P}otts model on signed networks,'' \emph{Chaos: An Interdisciplinary Journal of Nonlinear Science}, vol.~32, no.~8, 2022.

\bibitem{ojer2023modeling}
J.~Ojer, M.~Starnini, and R.~Pastor-Satorras, ``Modeling explosive opinion depolarization in interdependent topics,'' \emph{Phys. Rev. Lett.}, vol. 130, no.~20, p. 207401, 2023.

\bibitem{noutsos2006perron}
D.~Noutsos, ``On {P}erron--{F}robenius property of matrices having some negative entries,'' \emph{Linear Algebra Appl.}, vol. 412, no. 2-3, pp. 132--153, 2006.

\bibitem{horn2012matrix}
R.~A. Horn and C.~R. Johnson, \emph{Matrix Analysis}.\hskip 1em plus 0.5em minus 0.4em\relax Cambridge University Press, 2012.

\bibitem{zaslavsky2013matrices}
T.~Zaslavsky, ``Matrices in the theory of signed simple graphs,'' \emph{Proc. ICDM 2008, RMS-Lecture Notes Series}, no.~13, pp. 207--229, 2010.

\bibitem{belardo2019open}
F.~Belardo, S.~M. Cioab{\u{a}}, J.~H. Koolen, and J.~Wang, ``Open problems in the spectral theory of signed graphs,'' \emph{Art Discrete Appl. Math.}, vol.~1, no. P2.10, pp. 1--23, 2018.

\bibitem{zaslavsky1982signed}
T.~Zaslavsky, ``Signed graphs,'' \emph{Discrete Appl. Math.}, vol.~4, no.~1, pp. 47--74, 1982.

\bibitem{ghorbani2020sign}
E.~Ghorbani, W.~H. Haemers, H.~R. Maimani, and L.~P. Majd, ``On sign-symmetric signed graphs,'' \emph{arXiv preprint arXiv:2003.09981}, 2020.

\bibitem{altafini2014predictable}
C.~Altafini and G.~Lini, ``Predictable dynamics of opinion forming for networks with antagonistic interactions,'' \emph{IEEE Trans. Autom. Control}, vol.~60, no.~2, pp. 342--357, 2014.

\bibitem{jiang2016sign}
Y.~Jiang, H.~Zhang, and J.~Chen, ``Sign-consensus of linear multi-agent systems over signed directed graphs,'' \emph{IEEE Trans. Ind. Electron.}, vol.~64, no.~6, pp. 5075--5083, 2016.

\bibitem{fontan2021role}
A.~Fontan and C.~Altafini, ``The role of frustration in collective decision-making dynamical processes on multiagent signed networks,'' \emph{IEEE Trans. Autom. Control}, vol.~67, no.~10, pp. 5191--5206, 2021.

\bibitem{bizyaeva2022switching}
A.~Bizyaeva, G.~Amorim, M.~Santos, A.~Franci, and N.~E. Leonard, ``Switching transformations for decentralized control of opinion patterns in signed networks: application to dynamic task allocation,'' \emph{IEEE Control Syst. Lett.}, vol.~6, pp. 3463--3468, 2022.

\bibitem{bizyaeva2022sustained}
A.~Bizyaeva, A.~Franci, and N.~E. Leonard, ``Sustained oscillations in multi-topic belief dynamics over signed networks,'' in \emph{Proc. ACC}.\hskip 1em plus 0.5em minus 0.4em\relax IEEE, 2023, pp. 4296--4301.

\bibitem{wang2020biased}
L.~Wang, Y.~Hong, G.~Shi, and C.~Altafini, ``A biased assimilation model on signed graphs,'' in \emph{Proc. IEEE CDC}.\hskip 1em plus 0.5em minus 0.4em\relax IEEE, 2020, pp. 494--499.

\bibitem{converse2006nature}
P.~E. Converse, ``The nature of belief systems in mass publics (1964),'' \emph{Crit. Rev.}, vol.~18, no. 1-3, pp. 1--74, 2006.

\bibitem{leung2008psychological}
K.~Leung and M.~H. Bond, \emph{Psychological aspects of social axioms: Understanding global belief systems}.\hskip 1em plus 0.5em minus 0.4em\relax Springer Science \& Business Media, 2008.

\bibitem{golubitsky2009bifurcations}
M.~Golubitsky and R.~Lauterbach, ``Bifurcations from synchrony in homogeneous networks: linear theory,'' \emph{SIAM J. Appl. Dynamical Syst.}, vol.~8, no.~1, pp. 40--75, 2009.

\bibitem{stewart2011synchrony}
I.~Stewart and M.~Golubitsky, ``Synchrony-breaking bifurcation at a simple real eigenvalue for regular networks 1: 1-dimensional cells,'' \emph{SIAM J. Appl. Dynamical Syst.}, vol.~10, no.~4, pp. 1404--1442, 2011.

\bibitem{stewart2014synchrony}
I.~Stewart, ``Synchrony-breaking bifurcation at a simple real eigenvalue for regular networks 2: Higher-dimensional cells,'' \emph{SIAM J. Appl. Dynamical Syst.}, vol.~13, no.~1, pp. 129--156, 2014.

\bibitem{nijholt2019center}
E.~Nijholt, B.~Rink, and J.~Sanders, ``Center manifolds of coupled cell networks,'' \emph{SIAM Rev.}, vol.~61, no.~1, pp. 121--155, 2019.

\bibitem{golubitsky2023dynamics}
M.~Golubitsky and I.~Stewart, \emph{Dynamics and Bifurcation in Networks}.\hskip 1em plus 0.5em minus 0.4em\relax SIAM, 2023.

\bibitem{matcont}
A.~Dhooge, W.~Govaerts, Y.~Kuznetsov, H.~Meijer, and B.~Sautois, ``New features of the software {MatCont} for bifurcation analysis of dynamical systems,'' \emph{Math. Comp. Model. Dyn.}, vol.~14, no.~2, pp. 147--175, 2008.

\bibitem{Golubitsky1985}
M.~Golubitsky and D.~G. Schaeffer, \emph{Singularities and Groups in Bifurcation Theory}, ser. Applied Mathematical Sciences.\hskip 1em plus 0.5em minus 0.4em\relax New York, NY: Springer-Verlag, 1985, vol.~51.

\bibitem{fink1993oscillation}
E.~L. Fink, S.~Kaplowitz, G.~Barnett, and W.~Richards, ``Oscillation in beliefs and cognitive networks,'' \emph{Progress in Communication Sciences}, vol.~12, pp. 247--272, 1993.

\bibitem{fink2002oscillation}
E.~L. Fink, S.~A. Kaplowitz, and S.~M. Hubbard, ``Oscillation in beliefs and decisions,'' \emph{The Persuasion Handbook: Developments in Theory and Practice. Thousand Oaks, CA: Sage Publications}, pp. 17--38, 2002.

\bibitem{chung2012sequential}
S.~Chung, E.~L. Fink, L.~Waks, M.~F. Meffert, and X.~Xie, ``\edits{Sequential information integration and belief trajectories: An experimental study using candidate evaluations},'' \emph{Communication Monographs}, vol.~79, no.~2, pp. 160--180, 2012.

\bibitem{belleau1987cyclical}
B.~D. Belleau, ``Cyclical fashion movement: Women's day dresses: 1860-1980,'' \emph{Cloth. Text. Res. J.}, vol.~5, no.~2, pp. 15--20, 1987.

\bibitem{stimson2018public}
J.~Stimson, \emph{Public Opinion in America: Moods, Cycles, and Swings}.\hskip 1em plus 0.5em minus 0.4em\relax Routledge, 2018.

\bibitem{guckenheimer2013nonlinear}
J.~Guckenheimer and P.~Holmes, \emph{Nonlinear Oscillations, Dynamical Systems, and Bifurcations of Vector Fields}.\hskip 1em plus 0.5em minus 0.4em\relax New York, NY: Springer-Verlag, 2013, vol.~42.

\bibitem{AnastasiaThesis}
A.~S. Bizyaeva, ``Nonlinear dynamics of multi-agent multi-option belief and opinion formation,'' Ph.D. dissertation, Princeton University, 2022.

\bibitem{GolubitskySymmetryPerspective}
M.~Golubitsky and I.~Stewart, \emph{The Symmetry Perspective}, 1st~ed., ser. Progress in Mathematics.\hskip 1em plus 0.5em minus 0.4em\relax New-York: Birkh{\"a}user Basel, 2002, vol. 200.

\bibitem{Khalil2002}
H.~Khalil, \emph{Nonlinear Systems}, 3rd~ed.\hskip 1em plus 0.5em minus 0.4em\relax Prentice Hall, 2000.

\bibitem{li2023signed}
D.~Li and Q.~Huang, ``The signed graphs with symmetric spectra,'' \emph{arXiv:2304.06864 [math.CO]}, 2023.

\end{thebibliography}

\end{document}